\newif\ifisarxiv
\newif\ifisshort
\isarxivtrue
\isshortfalse

\ifisshort
    \newcommand{\fullver}[1]{}
    \newcommand{\confver}[1]{#1}
\else
    \newcommand{\fullver}[1]{#1}
    \newcommand{\confver}[1]{}
\fi

\ifisarxiv
    \newcommand{\nonarxivver}[1]{}
    \newcommand{\arxivver}[1]{#1}
\else
    \newcommand{\nonarxivver}[1]{#1}
    \newcommand{\arxivver}[1]{}
\fi

\ifisarxiv

\documentclass[10pt]{article}

\usepackage{setspace}
\usepackage{xcolor}
\usepackage[margin=1in]{geometry}
\usepackage{amssymb,amsthm,amsfonts,latexsym}
\usepackage{amsmath,bbm,xspace,graphicx,float,mathtools,epigraph}
\usepackage[colorlinks,citecolor=blue,bookmarks=true]{hyperref} 
\usepackage{enumitem,manyfoot,fullpage}
\usepackage{algorithm,algpseudocode}
\usepackage{comment}
\usepackage{soul}
\usepackage{subcaption}
\usepackage[compact]{titlesec}
\usepackage{booktabs}
\usepackage{cite}

\author{
    Komal Dhull \\
    Carnegie Mellon University \\
    \texttt{dhullkom@gmail.com} \\
    \and
    Steven Jecmen \\
    Carnegie Mellon University \\
    \texttt{sjecmen@cs.cmu.edu} \\
    \and
    Pravesh Kothari \\
    Carnegie Mellon University \\
    \texttt{praveshk@andrew.cmu.edu} \\
    \and
    Nihar B. Shah \\
    Carnegie Mellon University \\
    \texttt{nihars@cs.cmu.edu} \\
}
\date{}
\date{}

\else 

\relax
\documentclass[letterpaper]{article}
\usepackage{aaai22} 
\usepackage{times} 
\usepackage{helvet} 
\usepackage{courier} 
\usepackage[hyphens]{url} 
\usepackage{graphicx} 
\usepackage{graphicx}  
\usepackage{natbib}  
\usepackage{caption}  
\DeclareCaptionStyle{ruled}%
  {labelfont=normalfont,labelsep=colon,strut=off}
\usepackage{xcolor}
\usepackage{amssymb,amsthm,amsfonts,latexsym}
\usepackage{amsmath,bbm,xspace,graphicx,float,mathtools,epigraph}
\usepackage{enumitem,manyfoot}
\usepackage{algorithm,algpseudocode}
\usepackage{comment}
\usepackage{soul}
\usepackage{subcaption}
\usepackage[compact]{titlesec}
\usepackage{booktabs}
\usepackage{cite}

\author {
    Komal Dhull,
    Steven Jecmen,
    Pravesh Kothari,
    Nihar B. Shah
}
\affiliations {
    Carnegie Mellon University\\
    dhullkom@gmail.com, sjecmen@cs.cmu.edu, praveshk@andrew.cmu.edu, nihars@cs.cmu.edu
}

\fi

\newtheorem{lemma}{Lemma}
\newtheorem{theorem}{Theorem}
\newtheorem{definition}{Definition}
\newtheorem{proposition}{Proposition}

\newcommand{\assigndeg}{k}
\newcommand{\reviewers}{\mathcal{A}}
\newcommand{\papers}{\mathcal{P}}
\newcommand{\similarity}{s} 
\newcommand{\simwhole}{S}
\newcommand{\assignment}{\mathcal{M}}
\newcommand{\opt}{\textsf{Opt}}
\newcommand{\dirsimilarity}{G_{\optassign}}
\newcommand{\numrev}{n}
\newcommand{\numpap}{n}
\newcommand{\adrev}{a}
\newcommand{\adpap}{p}
\newcommand{\optassign}{\assignment^*}
\newcommand{\strategyproof}{strategyproof-via-partitioning}
\newcommand{\authorship}{\mathcal{U}}

\DeclareMathOperator*{\argmax}{argmax}

\makeatletter
\renewcommand*{\@fnsymbol}[1]{\ifcase#1\or \else\@arabic{\numexpr#1-1\relax}\fi}
\makeatother

\title{Strategyproofing Peer Assessment via Partitioning: \\ The Price in Terms of Evaluators' Expertise\arxivver{\thanks{The author ordering is alphabetical due to coincidence and not by policy~\cite{blog2019alphabetical}.}}}

\begin{document}

\maketitle

\begin{abstract} 
Strategic behavior is a fundamental problem in a variety of real-world applications that require some form of peer assessment, such as peer grading of homeworks, grant proposal review, conference peer review of scientific papers, and peer assessment of employees in organizations. Since an individual's own work is in competition with the submissions they are evaluating, they may provide dishonest evaluations to increase the relative standing of their own submission. This issue is typically addressed by partitioning the individuals and assigning them to evaluate the work of only those from different subsets. Although this method ensures strategyproofness, each submission may require a different type of expertise for effective evaluation. In this paper, we focus on finding an assignment of evaluators to submissions that maximizes assigned evaluators' expertise subject to the constraint of strategyproofness. We analyze the price of strategyproofness: that is, the amount of compromise on the assigned evaluators' expertise required in order to get strategyproofness.  We establish several polynomial-time algorithms for strategyproof assignment along with assignment-quality guarantees. Finally, we evaluate the methods on a dataset from conference peer review. 
\end{abstract}

\section{Introduction}

Many applications require evaluation of certain submissions. When the number of submissions is large enough to make independent expert evaluations of all of them infeasible, the individuals who submitted are each asked to evaluate submissions made by their peers.  
In education, peer grading of homeworks has become increasingly prevalent in Massive Open Online Courses (MOOCs)~\cite{shah2013case,diez2013peer,piech2013tuned} and conventional classrooms. In scientific research, peer review is used for grant proposals and conference paper submissions~\cite{shah2017design,tomkins2017reviewer,shah2021survey}. In the workplace, peer evaluation is frequently used to assess employee performance and determine employee promotions and bonuses~\cite{wexley1984perf,Fiore_Souza_2021}.

In many of these applications, peer assessment is \emph{competitive}, meaning that the eventual outcome of a submission is impacted by the evaluations of other submissions. Examples include a class graded on a curve such that only a certain percentage receives an `A' grade, a conference that intends to accept some fixed fraction of the papers, an agency awarding grants under a certain budget, or a company with a limited number of promotions on offer. 

A key challenge in competitive peer assessment is that agents behave \emph{strategically}: an agent may give low scores to the submissions they evaluate, in the hope that by hurting the chances of those submissions, they increase the relative chance of a good outcome for their own submission. A controlled experiment~\cite{balietti2016peer} found that people indeed behave in such a strategic manner in competitive peer assessment. 
Furthermore, the work~\cite{thurner2011peer} shows that even a small fraction of agents behaving strategically in scientific peer review can significantly lower the average quality of the accepted papers. It is thus vital to ensure the fairness and integrity of the process by developing mechanisms to prevent such strategic behavior. 
In fact, the NSF briefly experimented with a method (introduced by~\cite{merrifield2009telescope}) that attempts to prevent strategic behavior in the peer review of research proposals~\cite{naghizadeh2013incentive}, but this method does not come with theoretical guarantees.

By far the most well-studied way of ensuring strategyproofness is the partitioning method introduced in~\cite{alon2011sum} and studied further in~\fullver{\cite{holzman2013impartial,bousquet2014near,fischer2015optimal,aziz2016strategyproof,aziz2019strategyproof,mattei2020peernomination,kahng2018ranking,xu2018strategyproof}}\confver{\cite{holzman2013impartial,bousquet2014near,fischer2015optimal,xu2018strategyproof,aziz2019strategyproof}}. 
Under the partitioning method, submissions are partitioned into some number of subsets, and no agent is assigned a submission from the same subset as their own. 
The individual agent evaluations are then aggregated separately for each subset, so that any agent's evaluations cannot influence the final outcome for their own submission.

Apart from strategyproofness, another key aspect in assigning evaluators to submissions is matching based on expertise. For instance, in peer review of papers or proposals, not all agents have expertise for all papers or proposals. 
Similarly, in peer assessment within an organization, the peer assessors for any employee must be chosen to have a suitable understanding of that employee's work. In peer grading of essays or projects, the assessors must have the relevant background to do a suitable evaluation. Since the goal of peer assessment is to evaluate each submission as competently as possible, it is important to ensure that each submission is assigned evaluators with suitable expertise, or in other words, to maximize the quality of the assignment of evaluators to submissions. 

As both strategyproofness and assignment quality are crucial in many applications, \emph{our work studies the problem of finding a strategyproof assignment with maximum assignment quality}. The key question we ask is: \emph{what is the price paid by strategyproofing in terms of the assigned evaluators' expertise?} As a metric of evaluation, we use the ratio of the quality we obtain with strategyproofness to the maximum quality achievable without the strategyproofness constraint. 

Our work contributes to the body of literature on analyzing the price of strategyproofness in various settings~\cite{procaccia2013approximate,dughmi2010truthful,koutsoupias2014scheduling,ashlagi2015mix,kahng2018ranking}. This includes a line of work on impartial peer nomination/selection~\cite{alon2011sum,bousquet2014near,holzman2013impartial,aziz2016strategyproof,kurokawa2015impartial,fischer2015optimal,aziz2019strategyproof,mattei2020peernomination}, which focuses on selecting the best $k$ submissions in a strategyproof manner given an profile of evaluations. In contrast, we optimize the {\it assignment} of evaluators to submissions subject to a strategyproofness constraint and characterize the price of strategyproofness in terms of the assigned evaluators' expertise.  Further, our setting generalizes the standard peer selection setting, since evaluations may be used for various relative grading schemes other than best-$k$ selection. 
The prior work closest to ours is~\cite{xu2018strategyproof}, which considers the  partitioning mechanism specifically for conference peer review. They provide an algorithm that utilizes partitioning and conduct empirical analysis on its quality. However, their algorithm is designed to guarantee that the output ranking of submissions satisfies an efficiency property and does not focus on optimizing the evaluator assignment.

With that background, we now list {\bf our main contributions}:
\ifisshort \else \begin{enumerate} \fi 
    \fullver{\item}\confver{(1)} We establish fundamental limits on the amount of compromise that must be made on the assignment quality in order to impose strategyproofness via partitioning.
    \fullver{\item}\confver{(2)} We present polynomial-time computable algorithms that are optimal in the worst case.
    \fullver{\item}\confver{(3)} We show that the problem of instance-wise optimal strategyproof assignment via partitioning is NP-hard.
    \fullver{\item}\confver{(4)} We conduct experimental evaluations on data from the peer-review process of the ICLR 2018 conference, where we find that our algorithms achieve high-expertise assignments while producing fair partitions\fullver{ of papers}. 
\ifisshort \else \end{enumerate} \fi

\fullver{ We accomplish these goals using various techniques: applying combinatorial methods,  drawing a connection to equitable graph coloring, and formulating our problem as a max-cut problem. }

Apart from considering strategyproofness and assignment quality together, we note two points of contrast of our work as compared to the literature. First, previous works on strategyproof partitioning consider a uniform random partition in order to ensure fairness: that is, to ensure that no partition contains disproportionately strong or disproportionately weak submissions. In our work, we analyze the random partition approach and use it as a baseline for the rest of our results. Moreover, we conduct experiments using data from the ICLR conference, which reveal that the non-random partition output by our algorithms does not result in any substantial unfairness. Second, the  work~\cite{xu2018strategyproof}, which deals with both assignment quality and strategyproofing, considers arbitrary authorships where each submission may have multiple authors and each agent may have authored multiple submissions. In contrast, our theoretical analysis restricts attention to each agent having authored one submission and each submission being authored by one agent. 
Such one-to-one authorship occurs often in peer grading, peer assessment of employees, or peer review of certain proposals, and is equivalent to common settings in the strategyproofing literature~\cite{alon2011sum,bousquet2014near,holzman2013impartial,aziz2016strategyproof,kurokawa2015impartial,fischer2015optimal,aziz2019strategyproof,mattei2020peernomination,kahng2018ranking}. 
In Section~\ref{sec:gen_auth}, we further provide an extension and empirical evaluation that handles arbitrary authorships.

\arxivver{All of the code for our algorithms and our empirical results is freely available online at \url{https://github.com/sjecmen/optimal_strategyproof_assignment}.}
\nonarxivver{The full version of the paper can be found online,\footnote{\url{https://arxiv.org/abs/2201.10631}} as can all of the code for our algorithms and our empirical results.\footnote{\url{https://github.com/sjecmen/optimal_strategyproof_assignment}}}

\section{Background and Problem Formulation} \label{sec:prob}
We consider a setting of peer assessment between agents, where each agent first submits some work for evaluation and  is then assigned to evaluate other agents' submissions. After evaluations have been completed, submissions can be compared based on the evaluation scores in order to determine any competitive outcomes, such as relative grades (in a classroom setting), accept/reject decisions (in conference peer review), or employee bonuses and promotions (in an organization).

\subsection{Preliminaries} 
Let $\reviewers = \{\adrev_1, \dots, \adrev_\numrev\}$ be the set of agents and let $\papers = \{\adpap_1, \dots, \adpap_\numpap\}$ be the set of submissions from the agents. We assume that each agent $\adrev_i$ ($i \in [\numrev]$) authors exactly one submission $\adpap_i$. \fullver{(This is equivalent to common settings in the strategyproofing literature~\cite{holzman2013impartial,bousquet2014near,fischer2015optimal,aziz2016strategyproof,kahng2018ranking}. Furthermore, we handle arbitrary authorships in Section~\ref{sec:gen_auth}.) }

A key focus of our work is the assignment of agents to submissions for review. Constructing a high-quality assignment for peer assessment (in the absence of strategyproofing requirements) is a well-studied problem\fullver{, and is conducted in two phases. The first phase involves}\confver{. This involves first} computing a ``similarity'' between every agent-submission pair, a number between 0 and 1 where a higher value indicates a better match in terms of expertise. Similarities are computed in various ways~\fullver{\cite{charlin13tpms,mimno2007expertise,fiez2020super,meir2020market}}\confver{\cite{charlin13tpms,mimno2007expertise,fiez2020super}}. Our work is agnostic to the method used to compute similarity scores. We assume we are given a matrix $\simwhole \in [0, 1]^{\numrev \times \numrev}$ of `similarity scores' for each agent-submission pair that capture the expertise of each agent to evaluate each submission. For any $i \in [\numrev], j \in [\numpap]$, the $(i,j)^\text{th}$ entry of matrix $\simwhole$, denoted by $\similarity_{i,j}$, represents the similarity between 
agent $\adrev_i$ and submission $\adpap_j$, where a higher value means that one expects a better quality of evaluation.

\subsection{Assignments}

\fullver{The second phase of the assignment process then uses the similarities to assign submissions to agents. }For a predefined value $\assigndeg \in \mathbb{Z}_+$, an assignment with loads of $\assigndeg$ is defined as a set $\assignment \subseteq \reviewers \times \papers$ of assigned agent-submission pairs where each submission is assigned exactly $\assigndeg$ agents, each agent is assigned to exactly $\assigndeg$ submissions, and no agent is assigned to their own submission. It is important to note that in our applications of interest, the ``load'' $\assigndeg$ is typically a small constant independent of $\numrev$, and we will assume so throughout this paper. 

The assignment is chosen by maximizing a specified objective subject to the load constraints. By far the most common choice of objective is to maximize the sum of the assigned similarities~\confver{\cite{goldsmith2007ai,taylor2008optimal,charlin13tpms}}\fullver{\cite{goldsmith2007ai,taylor2008optimal,tang2010expert,charlin13tpms,charlin2011framework,li2016new}}, 
and this approach is widely used in practice: for instance, in IJCAI, NeurIPS, AAAI and other conferences. Formally, for any assignment $\assignment \subseteq \reviewers \times \papers$, the total similarity is given by 
    $\sum_{(\adrev_i, \adpap_j) \in \assignment} \similarity_{i,j}$.  
Fixing some $\assigndeg$, define $\optassign_\simwhole$ as the maximum-similarity assignment 
\begin{subequations}
\label{EqnSumSim}
\begin{align}
    \optassign_\simwhole = \argmax_{\assignment \subseteq \reviewers \times \papers} &\sum_{(\adrev_i, \adpap_j) \in \assignment} \similarity_{i,j} \\
    \text{subject to } &\sum_{\adrev_i \in \reviewers} \mathbb{I}[(\adrev_i, \adpap_j) \in \assignment] = \assigndeg && \forall \adpap_j \in \papers \label{eq:lp_first} \\
    &\sum_{\adpap_j \in \papers} \mathbb{I}[(\adrev_i, \adpap_j) \in \assignment] = \assigndeg && \forall \adrev_i \in \reviewers \\
    &(\adrev_i, \adpap_i) \not\in \assignment && \forall \adrev_i \in \reviewers \label{eq:lp_last}.
\end{align}
\end{subequations}
The optimal assignment (without strategyproofness) $\optassign_\simwhole$ can be found efficiently via standard methods such as min-cost flow algorithms or linear programming. 
Let $\opt_\simwhole$ be the similarity of $\optassign_\simwhole$ (leaving dependence on $\assigndeg$ implicit in the notation); that is, $\opt_\simwhole$ is the maximum value of the aforementioned objective under the stated constraints. When unambiguous, the subscript $\simwhole$ may be omitted. 

While we consider the aforementioned popular objective in most of our analysis, we note that another objective that is sometimes used is the leximin or max-min fairness of the assignment~\confver{\cite{garg2010assigning,stelmakh2018assignment,kobren19localfairness}}\fullver{\cite{garg2010assigning,stelmakh2018assignment}}, which we examine in \confver{Appendix}\fullver{Section}~\ref{sec:minimax}.

\subsection{Strategyproofness via Partitioning}
Our goal in this paper is to find maximum-similarity {\it strategyproof} assignments. A strategyproof assignment is one in which no agent can improve the outcome of their own submission by changing the evaluation they provide. 

As introduced earlier, a standard method for constructing strategyproof assignments begins by partitioning the agents into two subsets. 
An assignment of agents to submissions is then found, where agents can only be assigned to submissions authored by agents in the other subset. After evaluations are completed, any relative grading (e.g., classroom grading or accept/reject decisions) is done independently within each subset. Thus, the evaluation provided by any agent cannot influence the final outcome of their own submission.

In this paper, we use the term ``\strategyproof{}'' specifically to describe assignments produced in this way.
\begin{definition} \label{def:sp}
An assignment $\assignment$ is {\bf \strategyproof{}} if there exists a partition of $\reviewers$ into two subsets $\reviewers_1, \reviewers_2$ such that 
\begin{subequations}
\begin{align}
    &(\adrev_i, \adpap_j) \not\in \assignment \qquad \forall \adrev_i, \adrev_j \in \reviewers_t; \forall t \in \{1, 2\} \label{eq:sp_const} \\
    &\reviewers_1 \cup \reviewers_2 = \reviewers; \quad \reviewers_1 \cap \reviewers_2 = \emptyset \label{eq:sp_last}.
\end{align}
\end{subequations}
\end{definition}
In Section~\ref{sec:multipart}, we extend this definition to \fullver{allow for partitioning into }more than two subsets. Our goal is to find a maximum-similarity \strategyproof{} assignment
\begin{align*}
    &\argmax_{\reviewers_1, \reviewers_2 \subseteq \reviewers; \assignment \subseteq \reviewers \times \papers} \sum_{(\adrev_i, \adpap_j) \in \assignment} \similarity_{i,j} \\
    &\text{subject to } \eqref{eq:lp_first} - \eqref{eq:lp_last}, \eqref{eq:sp_const}, \eqref{eq:sp_last}.
\end{align*}
If an assignment satisfies~\eqref{eq:sp_const} for some partition, we say that assignment ``respects'' the partition; we say that a pair $(\adrev_i, \adpap_j)$ respects the partition if $\adrev_i$ and $\adrev_j$ are in different subsets. Note that the load constraints imply $|\reviewers_1| = |\reviewers_2|$ for any feasible solution, so we assume that $\numrev$ is even in all of our results; we also assume that $\assigndeg \leq \frac{\numrev}{2}$ for feasibility.

Given a partition $(\reviewers_1, \reviewers_2)$, finding the maximum-similarity assignment can be done via standard methods by additionally disallowing any pairs violating constraint~\eqref{eq:sp_const}. Thus, the primary question we consider in this paper is how to optimally choose the partition in order to maximize the similarity of the resulting assignment.

\subsection{Evaluation Metric}
We evaluate a \strategyproof{} assignment algorithm in terms of the ratio between the similarity of the assignment it produces and $\opt_\simwhole$, the similarity of the optimal non-strategyproof assignment. Specifically, consider any assignment algorithm that, given input similarities $\simwhole$, produces a \strategyproof{} assignment denoted by $\assignment_\simwhole$. We evaluate its performance in terms of the worst-case input similarities as:
\begin{align*}
    \min_{\simwhole : \opt_\simwhole > 0} \frac{\sum_{(\adrev_i, \adpap_j) \in \assignment_\simwhole} \similarity_{i,j} }{\opt_\simwhole}.
\end{align*}

\section{Theoretical Results}
In this section, we present our main theoretical results.

\subsection{Baseline: Random Partitioning} 
We begin with a result that provides a simple baseline for comparison:  Algorithm~\ref{algo:rand} chooses a partition uniformly at random. This is the approach taken by most prior literature on partitioning-based mechanisms. 
It is easy to show that such a uniformly random partition can attain at least half of the optimal similarity.

\begin{algorithm}[t!]
\caption{Random Partition}
\label{algo:rand}
\begin{algorithmic}[1]
\Require $\reviewers$, $\papers$, $\simwhole$, $\assigndeg$
\State Sample $\reviewers_1$ uniformly at random from $\{ \reviewers' :  \reviewers' \subseteq \reviewers, |\reviewers'| = |\reviewers| /2 \}$
\State $\reviewers_2 \gets \reviewers \setminus \reviewers_1$
\State $\assignment \gets $ max-similarity assignment with loads $\assigndeg$ respecting $(\reviewers_1, \reviewers_2)$
\State \Return assignment $\assignment$ and partition $(\reviewers_1, \reviewers_2)$
\end{algorithmic}
\end{algorithm}

\begin{proposition}\label{prop:rand}
For any $\assigndeg$ and any $\simwhole$, Algorithm~\ref{algo:rand} finds a \strategyproof{} assignment with similarity at least $\frac{1}{2} \opt_\simwhole$ in expectation.
\end{proposition}
\newcommand{\proofproprand}{Since it is feasible to assign all pairs in $\optassign_\simwhole$ that respect the partition, Algorithm~\ref{algo:rand} achieves expected similarity \ifisarxiv
\begin{align*}
\mathbb{E}_{\assignment}\left[\sum_{(\adrev_i, \adpap_j) \in \assignment} \similarity_{i,j} \right] &\geq
\sum_{(\adrev_i, \adpap_j) \in \optassign_\simwhole} \similarity_{i,j} (\mathbb{P}[\adrev_i \in \reviewers_1, \adrev_j \in \reviewers_2]  + \mathbb{P}[\adrev_j \in \reviewers_1, \adrev_i\in \reviewers_2] ) \\
&=\sum_{(\adrev_i, \adpap_j) \in \optassign_\simwhole} \similarity_{i,j} \left( \frac{\numrev}{\numrev-1} \right) \frac{1}{2}  \\
&\geq \frac{1}{2} \opt_\simwhole.
\end{align*} \else \begin{align*}
&\mathbb{E}_{\assignment}\left[\sum_{(\adrev_i, \adpap_j) \in \assignment} \similarity_{i,j} \right] \\
& \qquad \geq
\sum_{(\adrev_i, \adpap_j) \in \optassign_\simwhole} \similarity_{i,j} (\mathbb{P}[\adrev_i \in \reviewers_1, \adrev_j \in \reviewers_2] \\ 
&\qquad \qquad \qquad  + \mathbb{P}[\adrev_j \in \reviewers_1, \adrev_i\in \reviewers_2] ) \\
&\qquad =\sum_{(\adrev_i, \adpap_j) \in \optassign_\simwhole} \similarity_{i,j} \left( \frac{\numrev}{\numrev-1} \right) \frac{1}{2}  \\
&\qquad \geq \frac{1}{2} \opt_\simwhole.
\end{align*}\fi}
\fullver{\begin{proof} \proofproprand \end{proof}}
\confver{The proof is provided in Appendix~\ref{apdx:rand}. }
Note that this bound on the expected performance of random partitioning is tight in the limit as $\numrev$ grows: in the worst-case over similarities, Algorithm~\ref{algo:rand} achieves exactly $\left(\frac{\numrev}{\numrev-1} \right) \frac{1}{2}  \opt$ similarity. This occurs when all agent-submission pairs assigned by $\optassign$ have similarity $1$, and all other pairs have similarity $0$.

\subsection{Worst-Case Upper Bound} 
Since $\frac{1}{2} \opt$ is easily attainable, the next natural question is: how much better is achievable? We establish an upper bound of $\frac{\assigndeg+1}{2\assigndeg + 1}\opt$ on the worst-case performance of any \strategyproof{} assignment algorithm.

\begin{theorem} \label{thm:ubound} 
For any $\assigndeg$ and any $\numrev$, there exist similarities $\simwhole$ for $\numrev$ agents such that no \strategyproof{} assignment has similarity greater than  $\frac{\assigndeg+1}{2\assigndeg+1}\opt_\simwhole$.  
\end{theorem}
\begin{proof}
Place the agents into groups of size $2\assigndeg + 1$, leaving any remaining agents out. Within each complete group, number the agents from $0$ to $2\assigndeg$. For all $i$ from $0$ to $2\assigndeg$, set the similarity of $\adrev_i$ and $\adpap_{i+1}, \dots, \adpap_{(i + 1 + \assigndeg) \mod 2 \assigndeg + 1}$ to $1$. Set all other similarities to $0$. 
On these similarities, $\optassign$ can assign every similarity-$1$ pair, for a total of $\assigndeg (2 \assigndeg + 1)$ per group. The optimal partition splits each group into subsets of size $\assigndeg$ and $\assigndeg + 1$, allowing at most $\assigndeg (\assigndeg + 1)$ similarity-$1$ pairs to be assigned in each group.
\end{proof}

\subsection{Cycle-Breaking Algorithm}  \label{sec:deg1}
In this section, we present a simple algorithm that meets the upper bound of Theorem~\ref{thm:ubound} when  $\assigndeg=1$.

Define a ``cycle" $\gamma$ of length $\ell$ in an assignment as an ordered list of indices $\gamma_1, \dots, \gamma_{\ell}$ such that agent $\adrev_{\gamma_i}$ is assigned to submission $\adpap_{\gamma_{i+1}}$ (defining $\gamma_{\ell + 1} = \gamma_1$). In any assignment with loads $\assigndeg=1$, the full set of indices $[\numrev]$ can be uniquely partitioned into such cycles, since each agent is assigned to one submission and each submission is assigned one agent.

\begin{algorithm}[t!]
\caption{Cycle-Breaking Algorithm}
\label{algo:deg1}
\begin{algorithmic}[1]
\Require Agents $\reviewers$, papers $\papers$, similarities $\simwhole$, load $\assigndeg$
\State $\widetilde{\assignment}^*_\simwhole \gets$ max-similarity assignment with loads $1$ 
\State $\reviewers_1 \gets \emptyset$; $\reviewers_2 \gets \emptyset$
\For {cycle $\gamma$ of length $\ell$ in $\widetilde{\assignment}^*_\simwhole$}
\State $y \gets \min_{i \in [\ell]} \similarity_{{\gamma_i},{\gamma_{i+1}}}$ 
\State $A \gets \emptyset$; $B \gets \emptyset$
\For {$i \in [\ell]$}
\State $j \gets y + i \mod \ell$
\If {$i$ odd}
\State $A \gets A \cup \{\adrev_{\gamma_{j}}\}$
\Else 
\State $B \gets B \cup \{\adrev_{\gamma_{j}}\}$
\EndIf
\EndFor
\If {$|\reviewers_1| \leq |\reviewers_2|$} \label{ln:add_to_set}
\State $\reviewers_1 \gets \reviewers_1 \cup A$; $\reviewers_2 \gets \reviewers_2 \cup B$
\Else
\State $\reviewers_1 \gets \reviewers_1 \cup B$; $\reviewers_2 \gets \reviewers_2 \cup A$
\EndIf \label{ln:end_add}
\EndFor
\State $\assignment \gets $ max-similarity assignment with loads $\assigndeg$  respecting $(\reviewers_1, \reviewers_2)$
\State \Return assignment $\assignment$ and partition $(\reviewers_1, \reviewers_2)$
\end{algorithmic}
\end{algorithm}

Algorithm~\ref{algo:deg1} works by splitting each cycle in the optimal $\assigndeg=1$ assignment across the partition in the way that maximizes similarity. 
The following theorem shows a lower bound on the similarity of the \strategyproof{} assignment produced by this algorithm when $\assigndeg=1$. 
\begin{theorem} \label{deg1thm}
When $\assigndeg=1$, for any $\simwhole$, Algorithm~\ref{algo:deg1} finds a \strategyproof{} assignment with similarity at least $\frac{2}{3}\opt_\simwhole$ in polynomial time.
\end{theorem}
\begin{proof}
$(\reviewers_1, \reviewers_2)$ is a partition of $\reviewers$ since each agent is included in exactly one cycle in $\widetilde{\assignment}^*_\simwhole$. Further, $|\reviewers_1| = |\reviewers_2|$ since agents are added to the partition to keep it as balanced as possible and we assume $\numrev$ is even. 

We bound the value of the returned assignment $\assignment$ when $\assigndeg=1$. 
By construction, at most one agent-submission pair in each cycle of $\widetilde{\assignment}^*_\simwhole$ does not respect the partition. Any cycle containing such a disallowed pair must be of length at least three, and the disallowed pair must have the minimum similarity among all assigned pairs in the cycle. Since it is feasible to assign all pairs in $\widetilde{\assignment}^*_\simwhole$ that respect the partition, the value of the \strategyproof{} assignment must be at least $\frac{2}{3} \opt_\simwhole$.

The partitioning step can be done in $O(\numrev)$ time, since each agent is considered once, and finding the two maximum-similarity matchings can be done with high probability in $\widetilde{O}(\numrev^3)$ time~\cite{brand2021minimum}.
\end{proof}
 
\subsection{Coloring Algorithm} \label{sec:color}
In this section, we present another algorithm for strategyproof peer assessment, which meets the upper bound of Theorem~\ref{thm:ubound} for any $\assigndeg$. 
The algorithm begins by constructing a directed graph $\dirsimilarity$ representing the optimal assignment $\optassign$. This graph contains one vertex $v_i$ for all $i \in [\numrev]$, and an edge $(v_i, v_j)$ if $(\adrev_i, \adpap_j) \in \optassign$. We then find an equitable coloring of this graph, which is defined as follows.
\begin{definition} \label{def:eq_color}
For any $\alpha \in \mathbb{Z}_+$, an \textbf{equitable $\alpha$-coloring} of a directed graph $G=(V, E)$ is a function $f : V \to [\alpha]$ such that $f(v_i) \neq f(v_j) \quad \forall (v_i, v_j) \in E$ and $|\{v : f(v) = x\}| - |\{v : f(v) = y\}| \leq 1 \quad \forall x, y \in [\alpha]$. 
\end{definition}

The following well-known result shows that an equitable coloring of limited size can be found in polynomial time. 
\begin{theorem}~\cite{hajnal1970proof,kierstead2010fast} \label{thm:hajnal}
A graph $G = (V, E)$ with maximum degree at most $\Delta$ has an equitable $\Delta + 1$-coloring that can be found in $O(\Delta |V|^2)$ time.
\end{theorem}

\begin{algorithm}[t!]
\caption{Coloring Algorithm}
\label{algo:conj}
\begin{algorithmic}[1]
\Require Agents $\reviewers$, papers $\papers$, similarities $\simwhole$, load $\assigndeg$
\State $\optassign_\simwhole \gets$ max-similarity assignment with loads $\assigndeg$ 
\State $\dirsimilarity \gets$ directed graph representing $\optassign_\simwhole$
\State $f \gets$ equitable $(2\assigndeg + 2)$-coloring of $\dirsimilarity$ \label{ln:color}
\For{$T \in \{T : T \subseteq [2\assigndeg+2], |T| = \assigndeg+1\}$} \label{ln:color_part}
\State $\reviewers_{T} \gets \{ \adrev_i : v_i \in V, f(v_i) \in T\}$
\State $\reviewers_{T}' \gets \{ \adrev_i : v_i \in V, f(v_i) \not\in T\}$
\State $x_T \gets \sum_{\adrev_i \in \reviewers_T, \adrev_j \in \reviewers_{T}'}  \similarity_{i,j} \mathbb{I}[(\adrev_i, \adpap_j) \in \optassign_\simwhole] $ \label{ln:val}
\EndFor
\State $T^* = \argmax_T x_T$
\State $\reviewers_1 \gets \reviewers_{T^*}$; $\reviewers_2 \gets \reviewers_{T^*}'$
\State $\assignment \gets $ max-similarity assignment with loads $\assigndeg$ respecting $(\reviewers_1, \reviewers_2)$ 
\State \Return assignment $\assignment$ and partition $(\reviewers_1, \reviewers_2)$
\end{algorithmic}
\end{algorithm}

Algorithm~\ref{algo:conj} uses this result as a subroutine to find an equitable $(2\assigndeg +2)$-coloring of $\dirsimilarity$. It then partitions the colors in the way that maximizes the total similarity of pairs in $\optassign$ split by the partition. 
The following result proves that this algorithm is worst-case optimal.
\begin{theorem} \label{thm:coloralg} 
For any $\assigndeg$ and any $\simwhole$, if $\numrev$ is divisible by $2\assigndeg+2$, Algorithm~\ref{algo:conj} finds a \strategyproof{} assignment with similarity at least $\frac{\assigndeg+1}{2\assigndeg+1} \opt_\simwhole$ in polynomial time.
\end{theorem}
\begin{proof}
\confver{The assumption that $\numrev$ is divisible by $2\assigndeg+2$ is needed to guarantee that $|\reviewers_1| = |\reviewers_2|$. }
Each vertex in $\dirsimilarity$ has in-degree and out-degree $\assigndeg$, so the maximum (total) degree is at most $2\assigndeg$. 
Therefore, Line~\ref{ln:color} can be implemented using Theorem~\ref{thm:hajnal} as a subroutine. 
\fullver{Further, since $\numrev$ is divisible by $2\assigndeg + 2$, all colors have exactly $\frac{\numrev}{2\assigndeg+2}$ vertices and so $|\reviewers_1| = |\reviewers_2|$. 

}Next, we bound the value of the returned assignment $\assignment$. 
Suppose we modify Line~\ref{ln:color_part} to choose $T$ uniformly at random from the set. Then, the expectation of $x_T$ in Line~\ref{ln:val} is 
    $\mathbb{E} \left[ x_T \right] 
    =  \sum_{(\adrev_i, \adpap_j) \in \optassign_\simwhole} \similarity_{i,j} \left(\frac{\assigndeg+1}{2(\assigndeg+1)-1}\right)
    = \frac{\assigndeg+1}{2\assigndeg+1} \opt_\simwhole.$
Therefore, $x_{T^*} \geq \frac{\assigndeg+1}{2\assigndeg+1} \opt_\simwhole$. 
Since it is feasible to assign all pairs whose similarity is counted in $x_{T^*}$,  the assignment $\assignment$ has similarity at least  $x_{T^*}$. 

Assuming $\assigndeg$ is constant, the time complexity of the partitioning step is dominated by the $O(\numrev^2)$ time taken to find the equitable coloring. Finding the two maximum-similarity matchings can be done with high probability in $\widetilde{O}(\numrev^3)$ time~\cite{brand2021minimum}. 
\end{proof}
\fullver{The assumption that $\numrev$ is divisible by $2\assigndeg+2$ is needed to guarantee that the partition is balanced. However, for arbitrary $\numrev$, the subsets of the partition differ in size by only $\assigndeg+1$ agents at most. If there are a small number of ``reserve'' agents who did not submit any work and are not used in $\optassign$, these agents can provide any evaluations needed for a feasible assignment. 
Since $\assigndeg$ is a small constant (often $\leq 3$), having access to enough reserve agents is likely not an issue in practice. For example, in a scientific peer review setting, many extra non-author reviewers are available; in a classroom setting, an instructor could grade the extra submissions.}

\subsection{Hardness} 
Although our algorithms are optimal on the worst-case input, one might hope for algorithms that can guarantee optimal performance on all inputs. However, the following result shows that when $\assigndeg \geq 2$, this is NP-hard.

\begin{theorem} \label{thm:nph}
For any $\assigndeg \geq 2$, it is NP-hard to find the optimal \strategyproof{} assignment, even when similarities are binary (that is,  when $\simwhole \in \{0, 1\}^{\numrev \times \numrev}$).
\end{theorem}
\begin{proof}[Proof Sketch]
The proof is by reduction from the ``Simple Max Cut on Cubic Graphs'' problem~\cite{yannakakis1978node}. We construct an instance of the \strategyproof{} assignment problem where each agent corresponds to a vertex. For some orientation of the input graph, we set $\similarity_{ij}=1$ for each directed edge $(v_i, v_j)$, and set similarities to zero elsewhere. These edges could all be assigned by $\optassign$ when $\assigndeg \geq 2$, but the optimal \strategyproof{} assignment is limited to the max-cut value in the original graph.
\end{proof}
The complete proof is provided in \arxivver{Appendix~\ref{apdx:nph}}\nonarxivver{Appendix~A}.

\begin{algorithm}[t!]
\caption{Multi-Partition Algorithm} 
\label{algo:multi}
\begin{algorithmic}[1]
\Require Agents $\reviewers$, papers $\papers$, similarities $\simwhole$, load $\assigndeg$
\State $\optassign_\simwhole \gets$ max-similarity assignment with loads $\assigndeg$ 
\State $\dirsimilarity \gets$ directed graph representing $\optassign_\simwhole$
\State $f \gets$ equitable $(2\assigndeg + 1)$-coloring of $\dirsimilarity$
\State \Return \fullver{assignment} $\optassign_\simwhole$ and \fullver{partition with $2\assigndeg+1$ subsets} $(\{\adrev_j : v_j \in V, f(v_j) = i\}_{i \in [2\assigndeg + 1]})$
\end{algorithmic}
\end{algorithm}

\subsection{Partitions With More Than Two Subsets} \label{sec:multipart}
We now relax the definition of ``\strategyproof{}'' given in Definition~\ref{def:sp}. Rather than requiring that agents be partitioned into two subsets, we allow them to be partitioned into any constant  (i.e., not depending on $\numrev$) number of subsets. This slight relaxation of our problem formulation allows us to obtain a \strategyproof{} assignment that achieves total similarity $\opt_\simwhole$ for any $\simwhole$.

\nonarxivver{
\begin{figure*}[t!] 
    \centering
    \begin{subfigure}{0.32\textwidth}\includegraphics[width=1\textwidth]{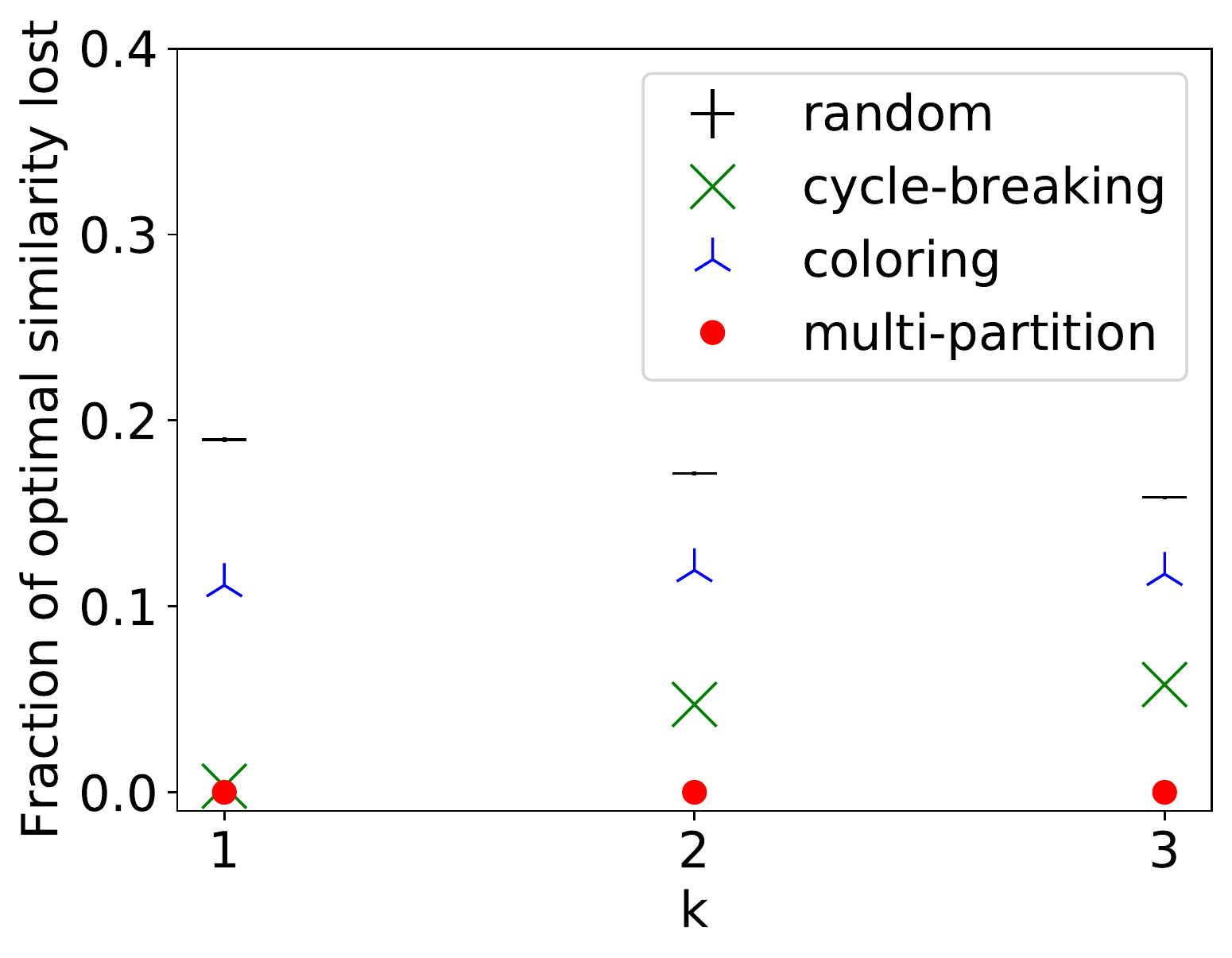}\caption{Assignment similarity lost}\label{fig:sim} \end{subfigure}
    \begin{subfigure}{0.32\textwidth}\includegraphics[width=1\textwidth]{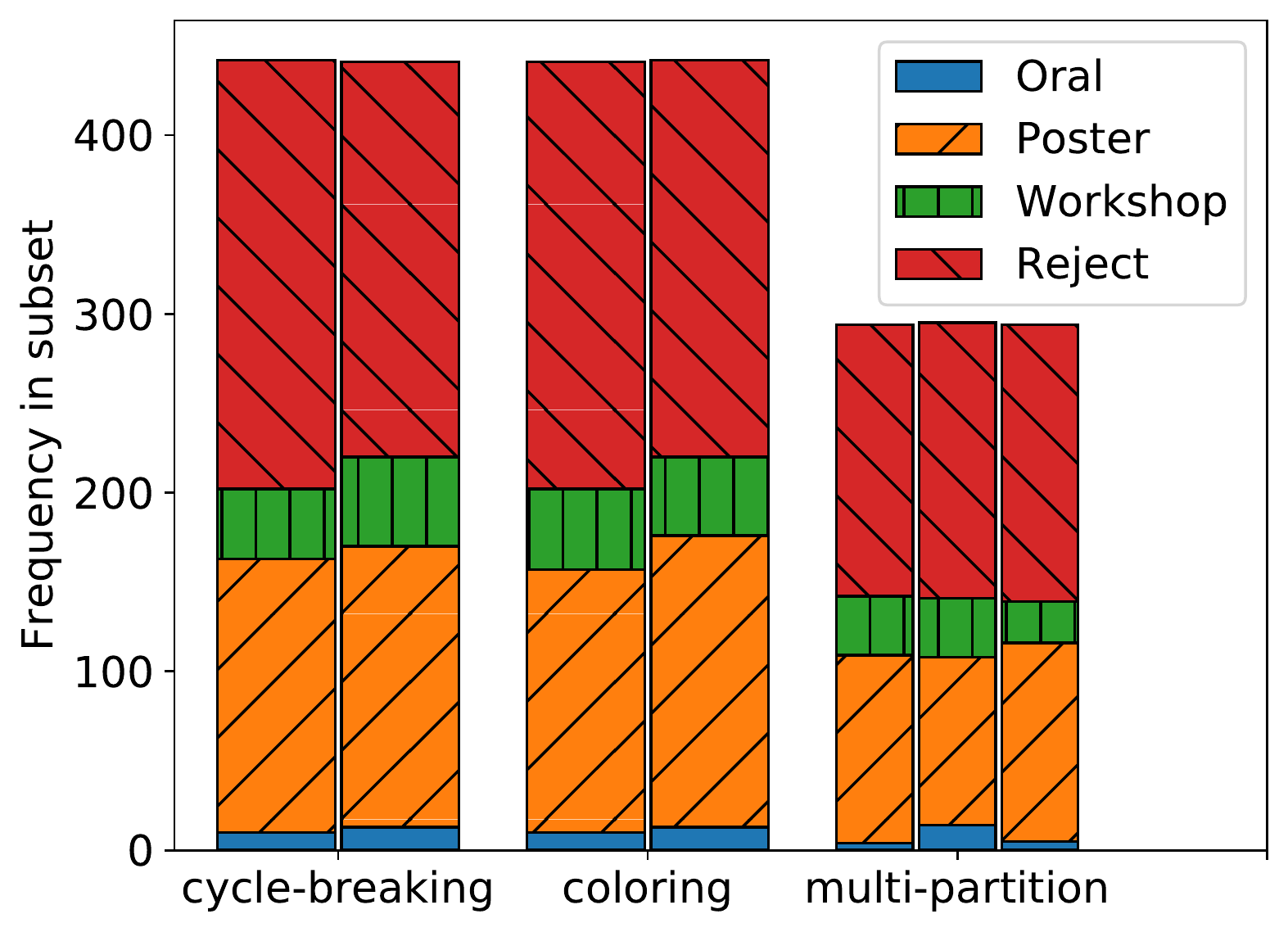}\caption{Partitioned paper decisions, \\ $\assigndeg=1$}\label{fig:outcome1} \end{subfigure} 
    \begin{subfigure}{0.32\textwidth}\includegraphics[width=1\textwidth]{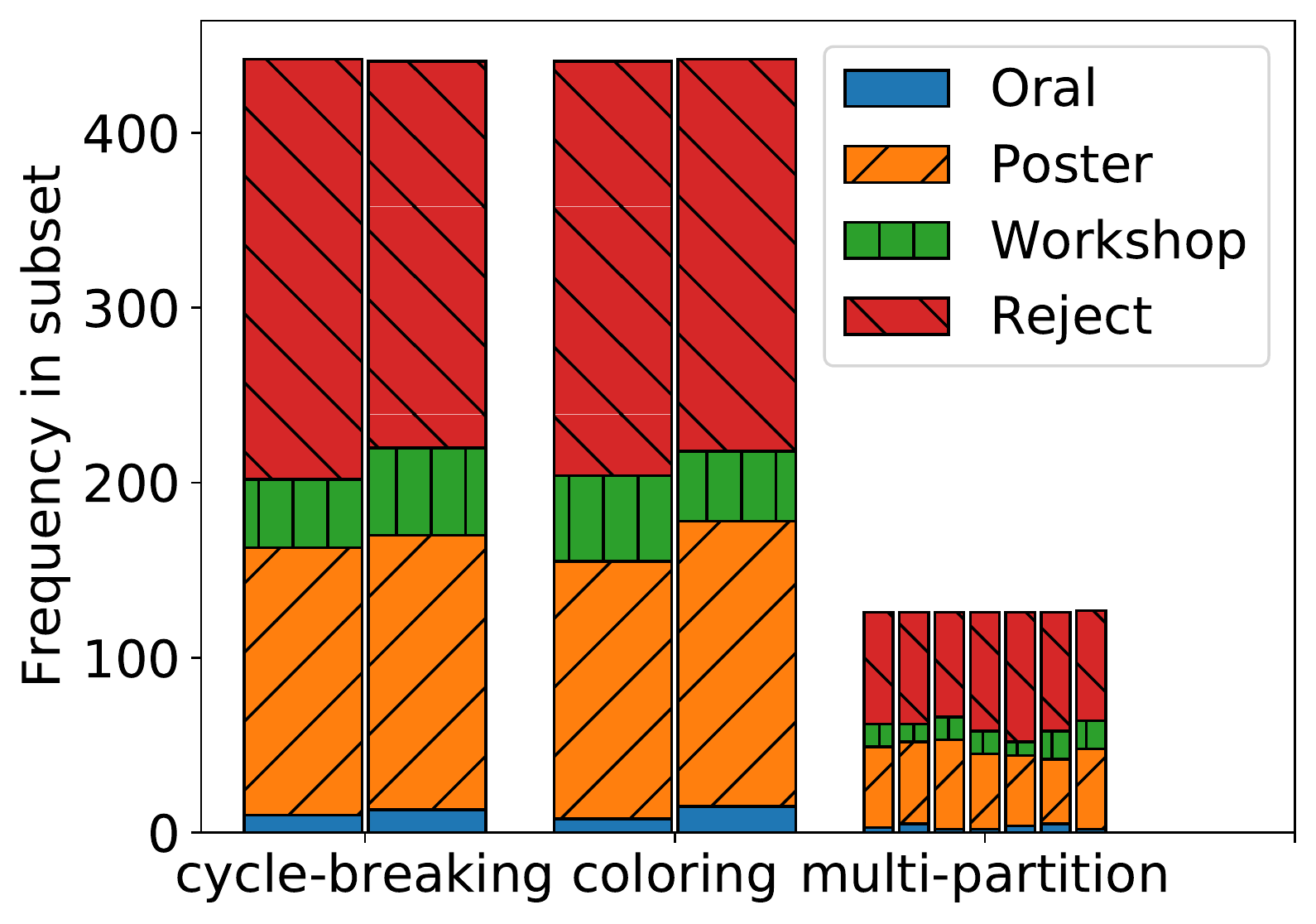}\caption{Partitioned paper decisions, \\ $\assigndeg=3$}\label{fig:outcome3} \end{subfigure} \\
    \begin{subfigure}{0.32\textwidth}\includegraphics[width=1\textwidth]{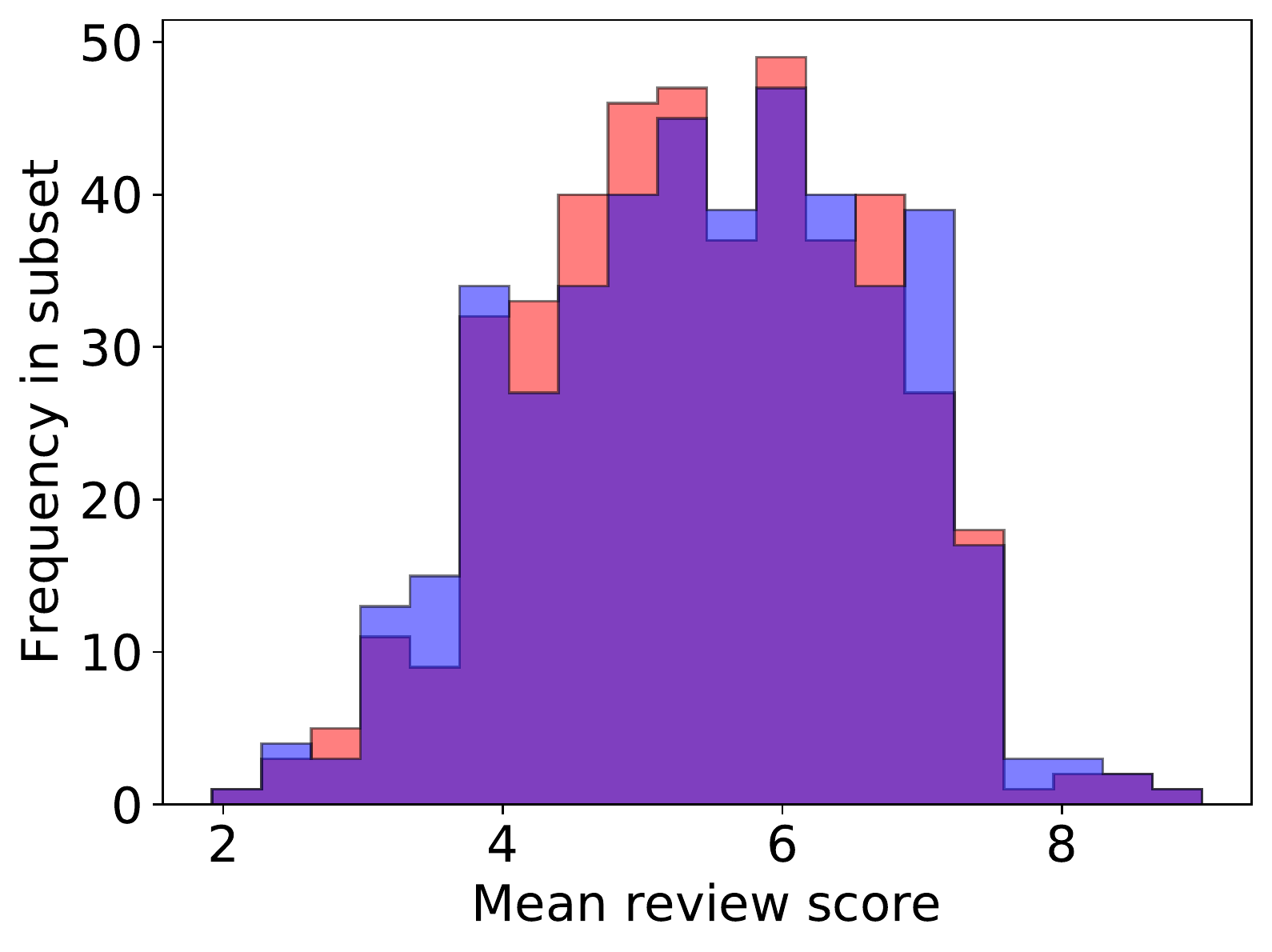}\caption{Partitioned paper scores for \\ the cycle-breaking algorithm, $\assigndeg=1$}\label{fig:score_cycle1} \end{subfigure}
    \begin{subfigure}{0.32\textwidth}\includegraphics[width=1\textwidth]{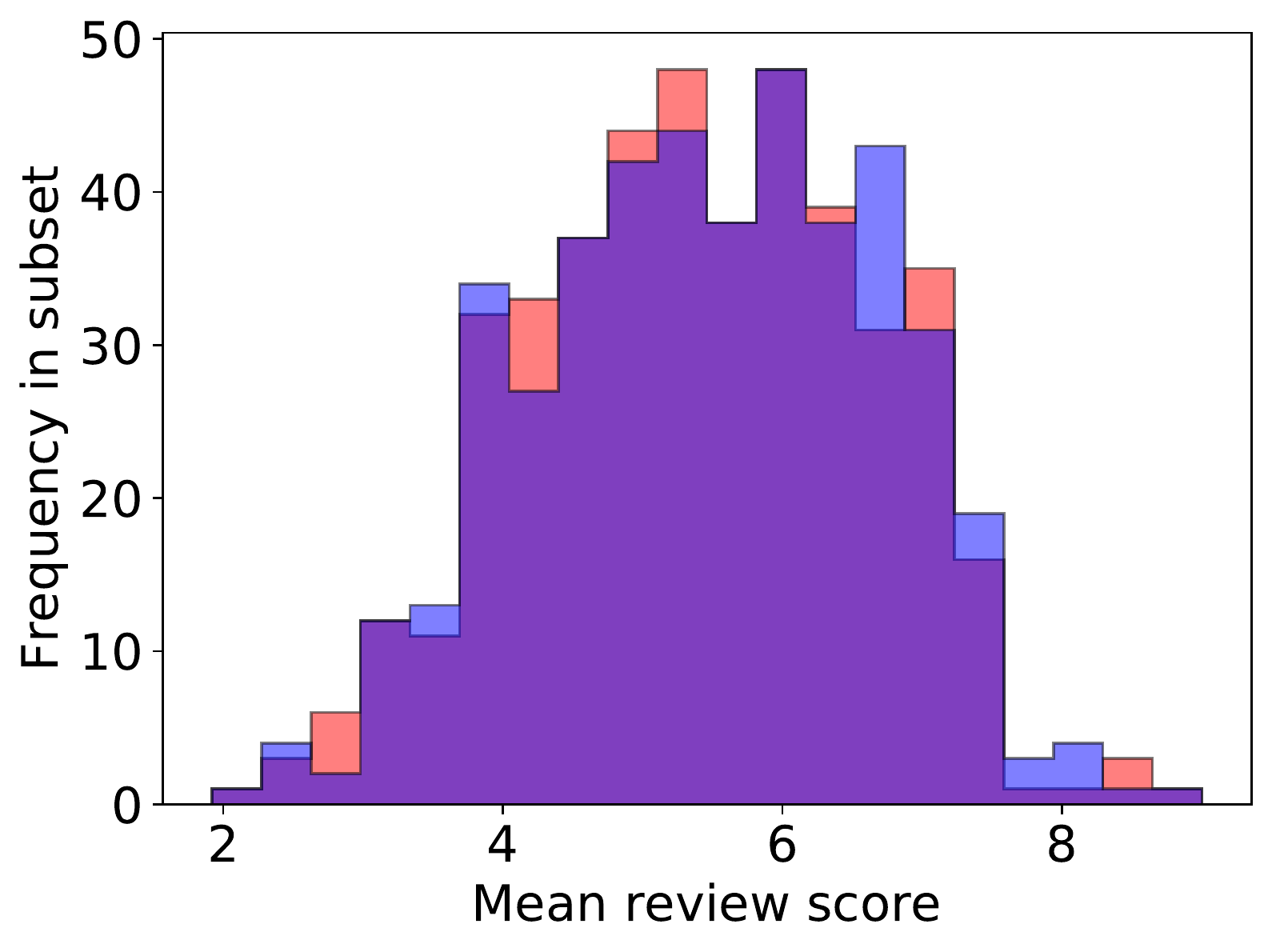}\caption{Partitioned paper scores for \\ the coloring algorithm, $\assigndeg=1$}\label{fig:score_color1} \end{subfigure} 
    \begin{subfigure}{0.32\textwidth}\includegraphics[width=1\textwidth]{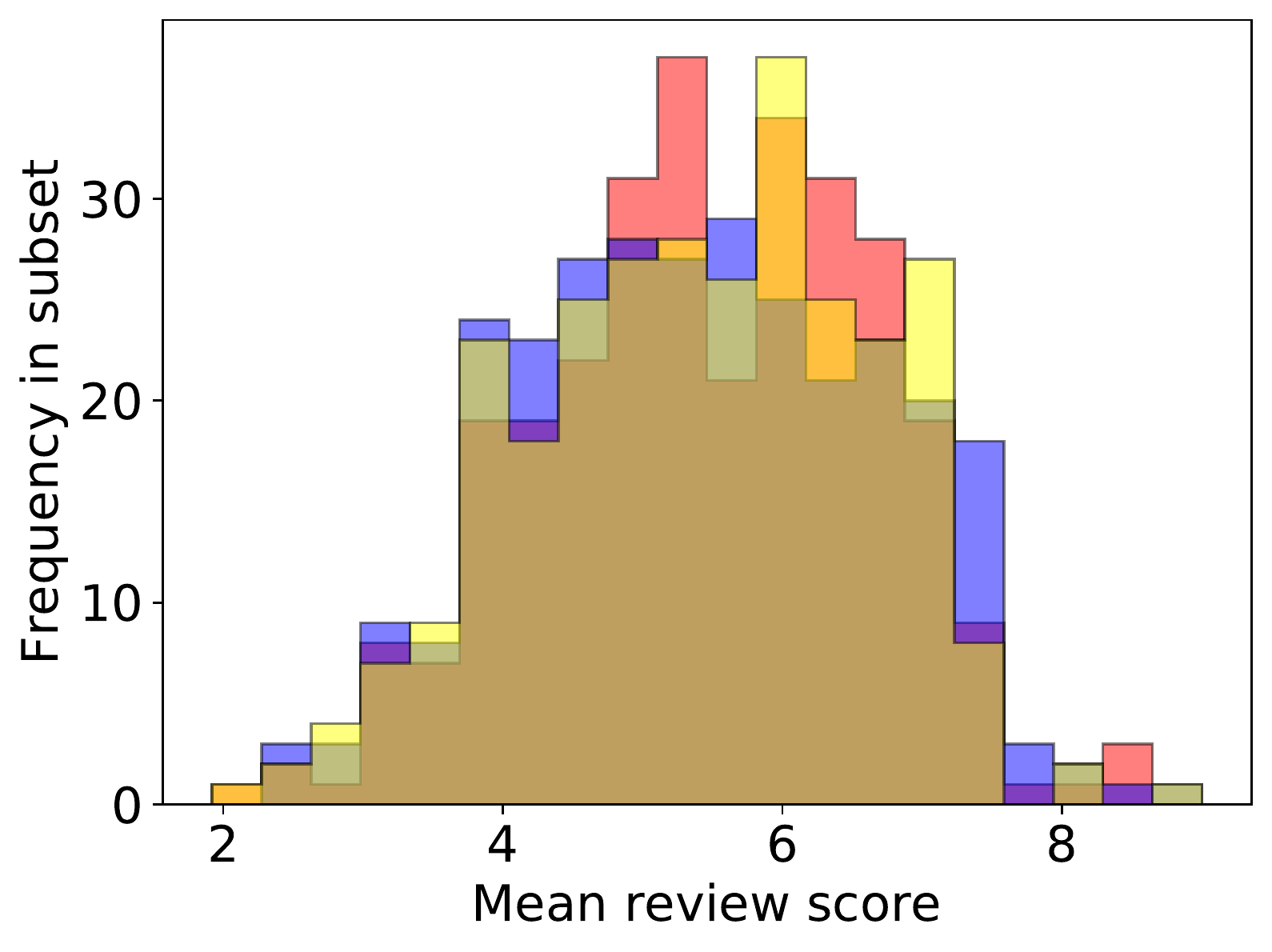}\caption{Partitioned paper scores for \\ the multi-partition algorithm, $\assigndeg=1$}\label{fig:score_multi1} \end{subfigure} 
    \caption{Experimental results on data from ICLR 2018.} \label{fig:results} 
\end{figure*}
}

\begin{theorem} \label{thm:kpart} 
For any $\assigndeg \geq 1$ and any $\simwhole$, Algorithm~\ref{algo:multi} finds a partition of agents into $2\assigndeg+1$ subsets, where each subset contains either $\lfloor \frac{\numrev}{2\assigndeg+1} \rfloor$ or $\lceil \frac{\numrev}{2\assigndeg+1} \rceil$ agents, and a \strategyproof{} assignment respecting this partition in polynomial time. This assignment has total similarity $\opt_\simwhole$.
\end{theorem}
\newcommand{\proofthmkpart}{Each vertex in $\dirsimilarity$ has in-degree and out-degree $\assigndeg$, so the maximum (total) degree is at most $2\assigndeg$. Therefore, by Theorem~\ref{thm:hajnal} we can find an equitable $(2\assigndeg + 1)$-coloring of $\dirsimilarity$ in $O(\numrev^2)$ time. By Definition~\ref{def:eq_color}, the entirety of $\optassign_\simwhole$ respects the partition induced by the coloring and so is \strategyproof{} with respect to this partition. Also by Definition~\ref{def:eq_color}, all color classes differ in size by at most $1$.}
\fullver{\begin{proof} \proofthmkpart \end{proof}}
Algorithm~\ref{algo:multi} constructs a directed graph representing $\optassign$ as described in Section~\ref{sec:color}. It then finds an equitable $(2\assigndeg+1)$-coloring using Theorem~\ref{thm:hajnal} and uses this coloring as the partition. 
\confver{The complete proof is presented in Appendix~\ref{apdx:kpart}. }

Although we can recover the entire optimal similarity with this method, increasing the number of subsets comes at the cost of reliability in determining the post-evaluation outcomes, since all relative outcomes must be chosen independently in each subset. In Section~\ref{sec:exps}, we experimentally examine this cost.

\fullver{\subsection{Fairness Objective} \label{sec:minimax}}
\newcommand{\secfairness}{So far we have analyzed the objective of maximizing total similarity~\eqref{EqnSumSim} due to its widespread use. However, this objective has been found to result in imbalanced or unfair assignments~\cite{stelmakh2018assignment}. An alternative proposed in the literature is to optimize the max-min fairness objective, which maximizes the total similarity assigned to the submission with minimum assigned similarity~\cite{garg2010assigning,stelmakh2018assignment,kobren19localfairness}. 
Formally, the problem of finding the optimal \strategyproof{} assignment under this objective is:
\begin{align*}
    &\argmax_{\reviewers_1, \reviewers_2 \subseteq \reviewers; \assignment \subseteq \reviewers \times \papers} \min_{\adpap_j \in \papers} \sum_{\adrev_i \in \reviewers} \similarity_{i,j} \mathbb{I}[(\adrev_i, \adpap_j) \in \assignment] \\
    &\text{subject to } \eqref{eq:lp_first} - \eqref{eq:lp_last}, \eqref{eq:sp_const}, \eqref{eq:sp_last}.
\end{align*}
Assignment algorithms optimizing this objective have been used in venues such as ICML 2020 and implemented in conference management platforms such as OpenReview.net.

In this section, we analyze the price of strategyproofing under this max-min objective. The following result shows that unfortunately, we cannot hope to do well on this objective in the worst-case.
\begin{theorem} \label{thm:minimax}
For any $\assigndeg$ and any $\numrev \geq 6$, there exist similarities $\simwhole$ on $\numrev$ agents such that the optimal non-strategyproof{} assignment has max-min objective value strictly greater than $0$ while no \strategyproof{} assignment has a max-min objective value greater than $0$.  
\end{theorem}
\begin{proof}
Split the agents into two groups such that both groups have an odd number of agents at least $3$; this is possible since we assume $\numrev$ is even. Within each group $\{\adrev_{\gamma_1}, \dots, \adrev_{\gamma_\ell}\}$ of size $\ell$, set similarities $\similarity_{\gamma_i, \gamma_{i+1}} = 1$ for all $i \in [\ell-1]$ and $\similarity_{\gamma_\ell, \gamma_1} = 1$. Set similarities to $0$ elsewhere. On these similarities, the optimal non-strategyproof{} assignment can assign every similarity-$1$ pair for a max-min fairness of $1$. However, since the number of reviewers in each group is odd, any partition of $\reviewers$ into two subsets must place two agents $\adrev_{\gamma_i}, \adrev_{\gamma_{i+1}}$ (or $\adrev_{\gamma_\ell}, \adrev_{\gamma_1}$) from each group in the same subset. Therefore, some submission $\adpap_{\gamma_{i+1}}$ from each group will have an assigned similarity of $0$.  
\end{proof}}
\fullver{\secfairness}

\arxivver{
\begin{figure*}[t!] 
    \centering
    \includegraphics[width=0.5\textwidth]{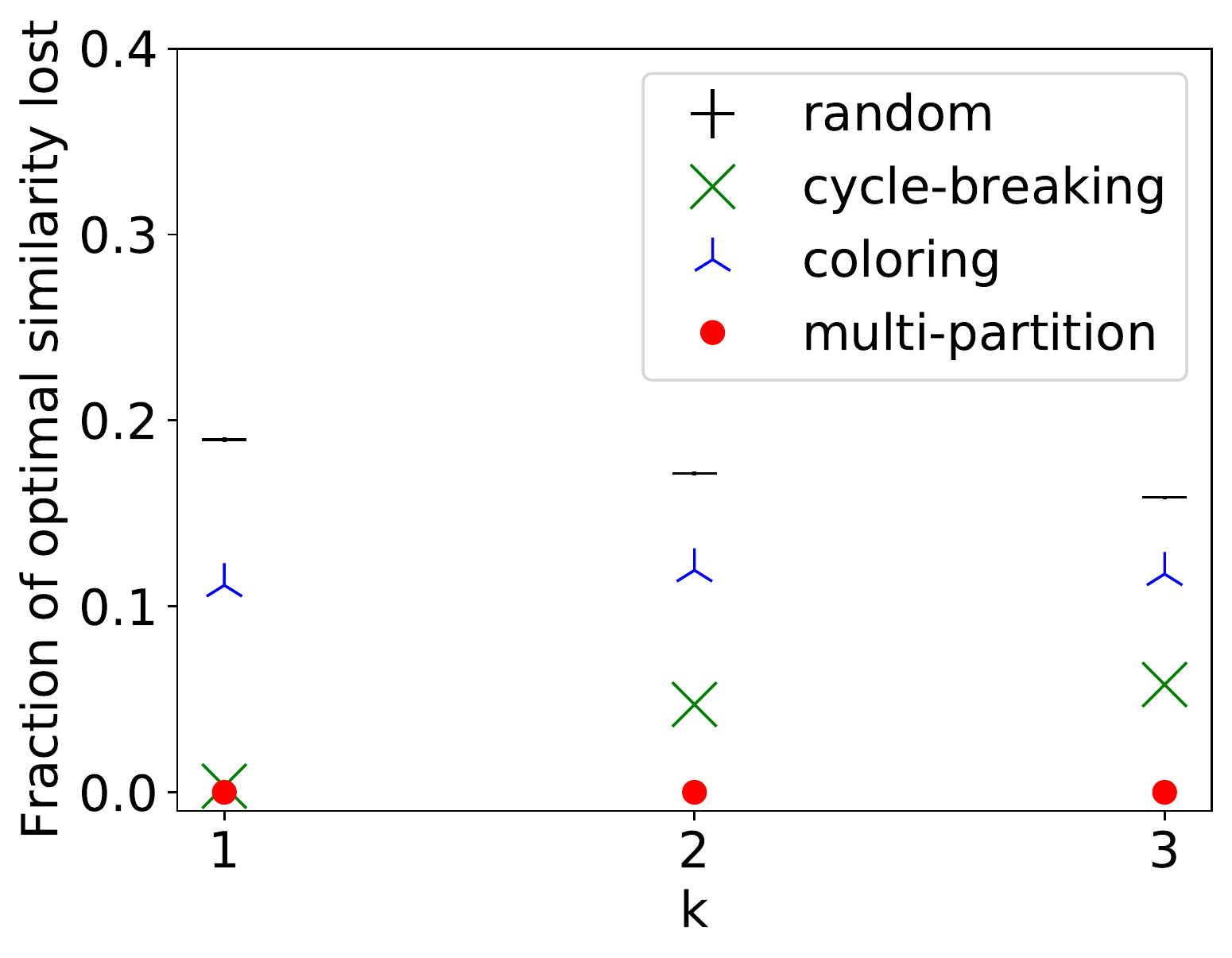}
    \caption{Assignment similarity lost on data from ICLR 2018.} \label{fig:sim} 
\end{figure*}

\begin{figure*}[t!] 
    \centering
    \begin{subfigure}{0.45\textwidth}\includegraphics[width=1\textwidth]{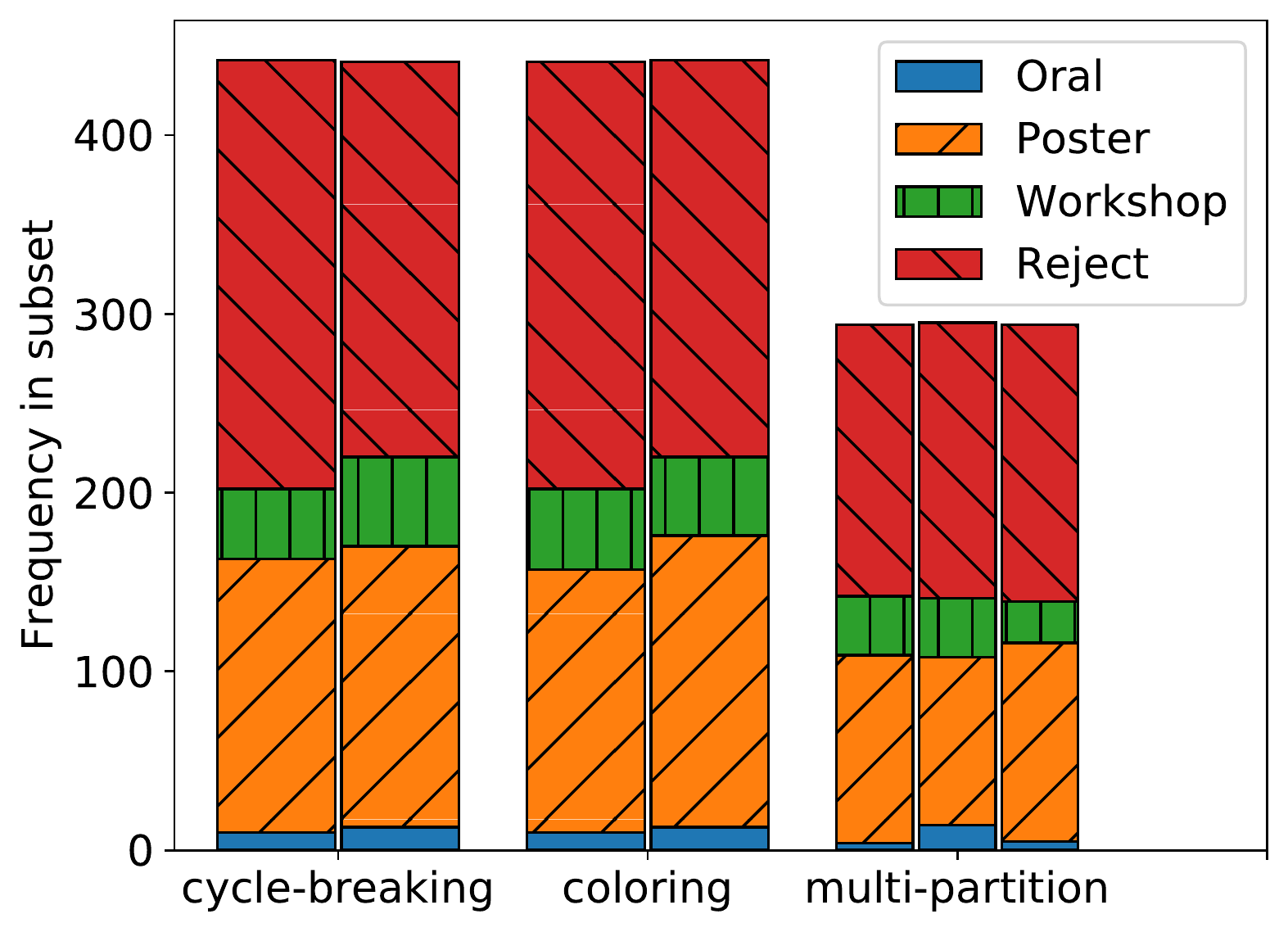}\caption{Partitioned paper decisions, \\ $\assigndeg=1$}\label{fig:outcome1} \end{subfigure} 
    \begin{subfigure}{0.45\textwidth}\includegraphics[width=1\textwidth]{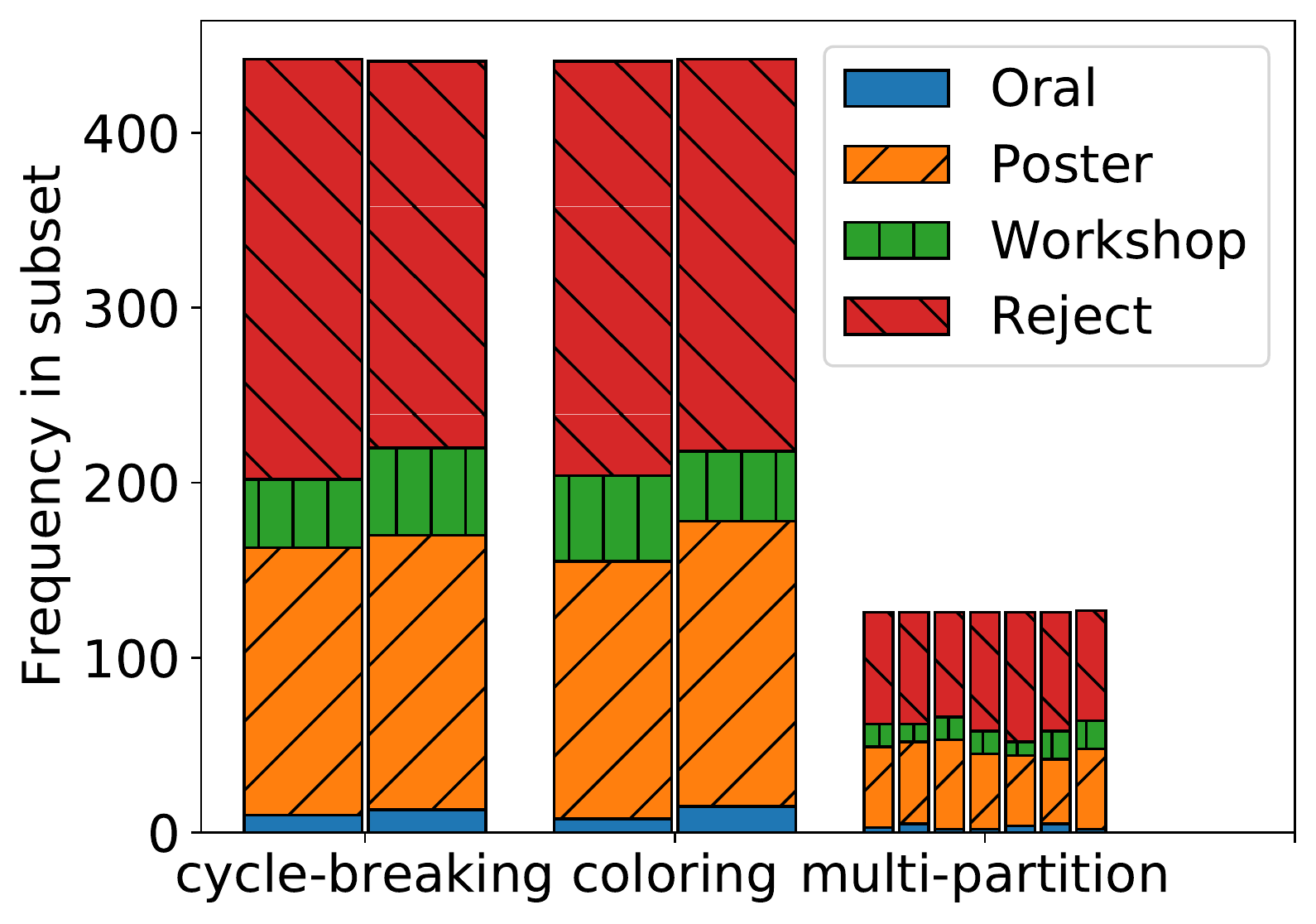}\caption{Partitioned paper decisions, \\ $\assigndeg=3$}\label{fig:outcome3} \end{subfigure} \\
    \begin{subfigure}{0.32\textwidth}\includegraphics[width=1\textwidth]{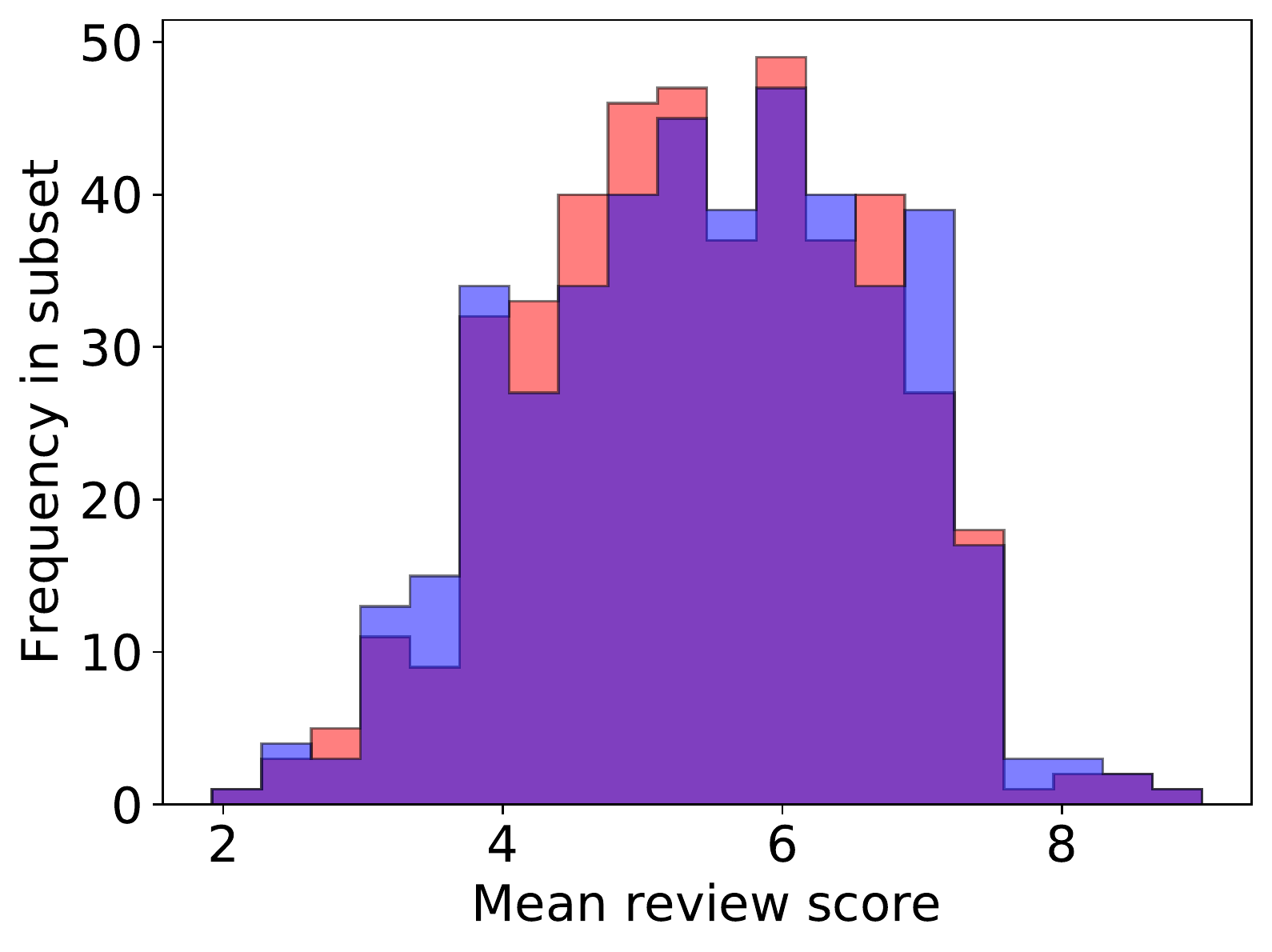}\caption{Partitioned paper scores for \\ the cycle-breaking algorithm, $\assigndeg=1$}\label{fig:score_cycle1} \end{subfigure}
    \begin{subfigure}{0.32\textwidth}\includegraphics[width=1\textwidth]{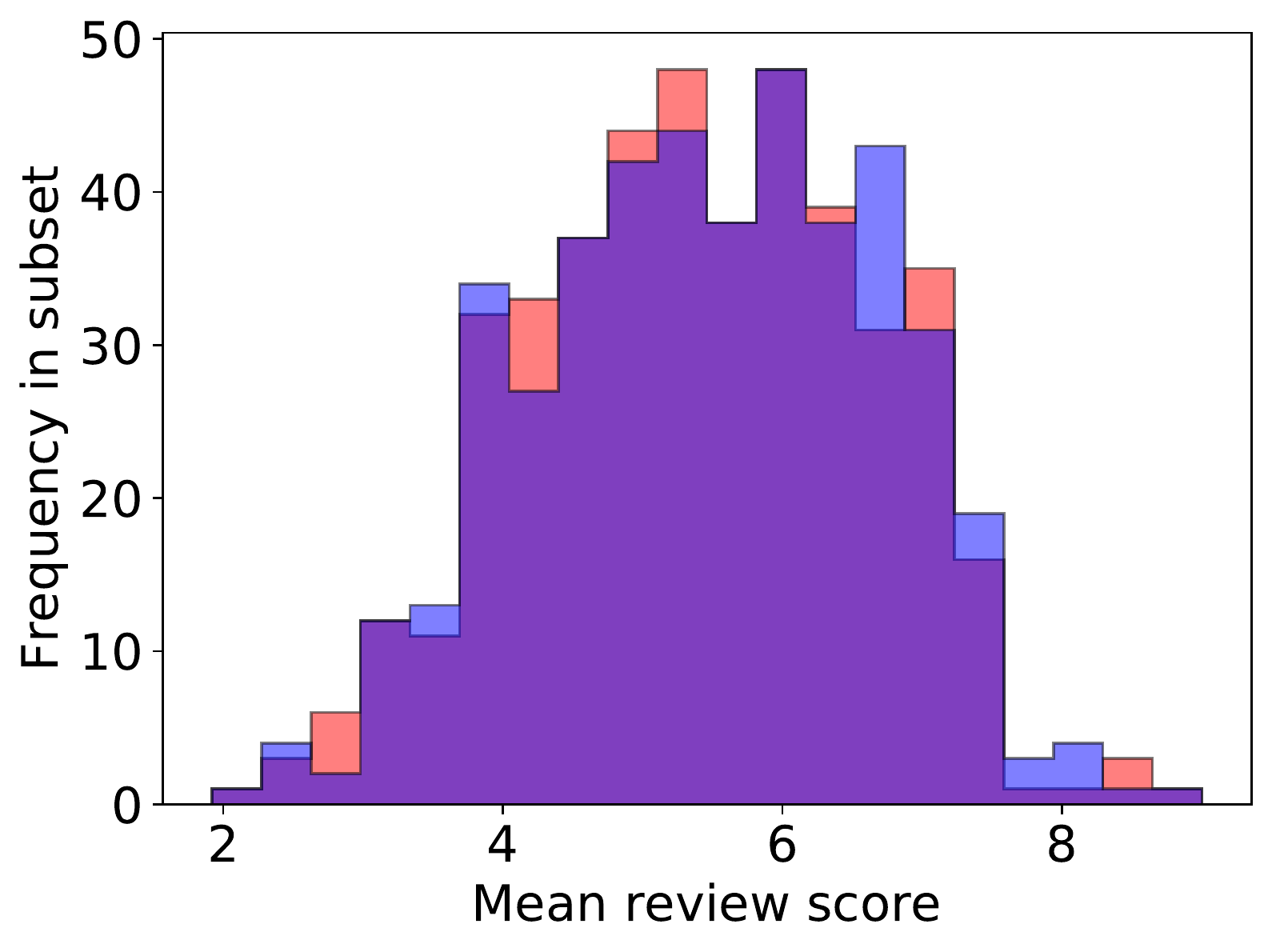}\caption{Partitioned paper scores for \\ the coloring algorithm, $\assigndeg=1$}\label{fig:score_color1} \end{subfigure} 
    \begin{subfigure}{0.32\textwidth}\includegraphics[width=1\textwidth]{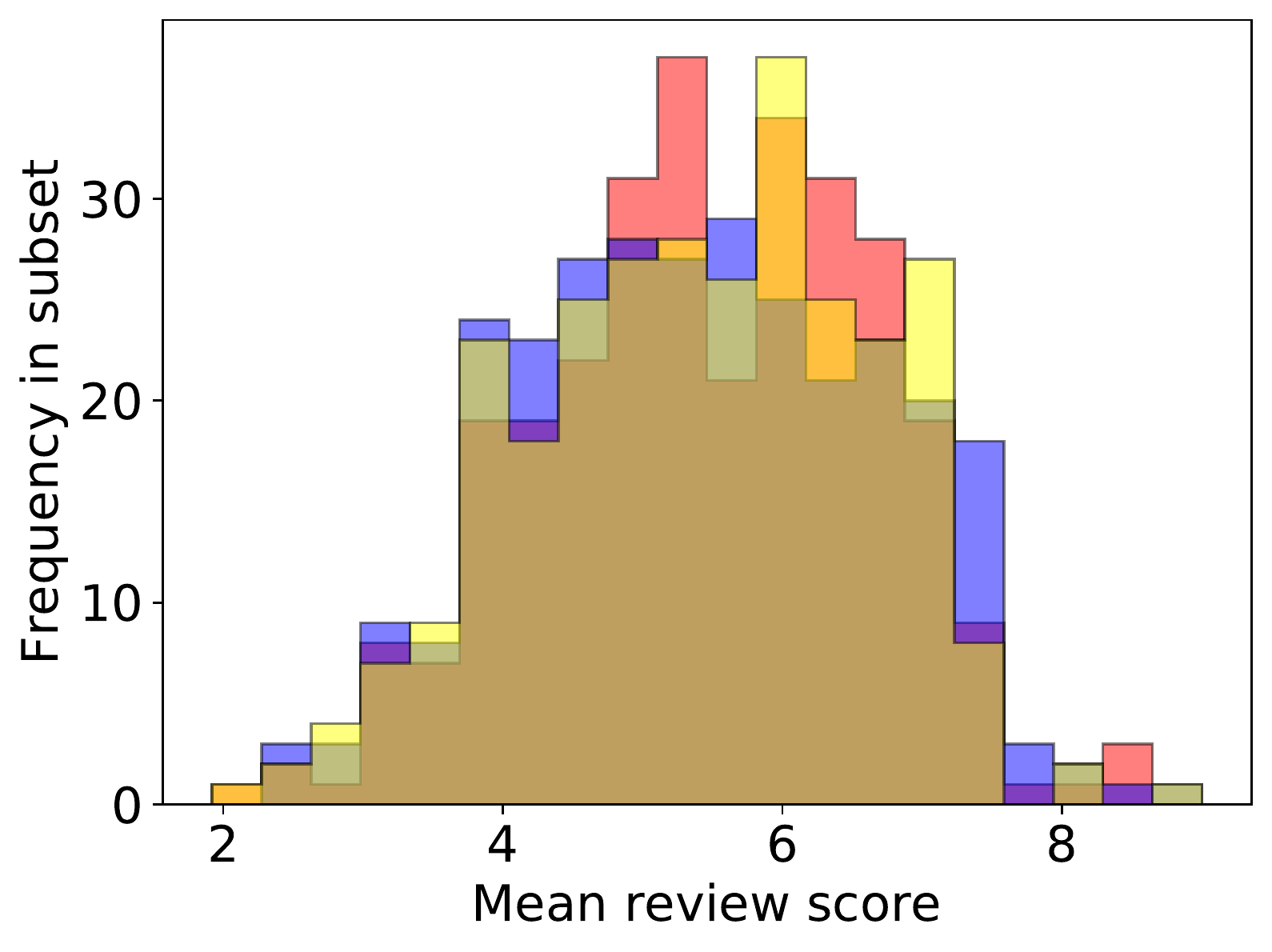}\caption{Partitioned paper scores for \\ the multi-partition algorithm, $\assigndeg=1$}\label{fig:score_multi1} \end{subfigure} 
    \caption{Partition quality on data from ICLR 2018.} \label{fig:results} 
\end{figure*}
}

\section{Experimental Results} \label{sec:exps}
In this subsection, we experimentally examine the performance of algorithms for \strategyproof{} assignment.

\subsection{Setup}
We evaluate our algorithms on data from the peer-review process at ICLR 2018. We use similarities recreated in~\cite{xu2018strategyproof}. 
To evaluate the partition quality, we also use the actual review scores and the accept/reject decisions at the ICLR 2018 conference~\cite{he2020openreview}.

Since our algorithms require that each agent authors exactly one submission, we find a maximum one-to-one matching on the real authorship graph and use this as the authorship for our experiments. This resulted in matching 883 out of the 911 papers. We then discarded any reviewers and papers not included in the authorship graph. 
Any additional reviewers required for feasibility (due to the divisibility of $\numrev$) have zero similarity with all papers. 

We evaluate four partitioning algorithms: random partitioning (Algorithm~\ref{algo:rand}), the cycle-breaking algorithm (Algorithm~\ref{algo:deg1}), the coloring algorithm (Algorithm~\ref{algo:conj}), and the multi-partition algorithm (Algorithm~\ref{algo:multi}). 
Since each paper received $3$ reviews at ICLR 2018, we test values of $\assigndeg \in \{1, 2, 3\}$.

\fullver{Additional experimental results are available in \arxivver{Appendix~\ref{apdx:exps}}\nonarxivver{Appendix~B}.}

\subsection{Assignment Similarity}
We first examine the similarity of the \strategyproof{} assignments produced by each algorithm. In Figure~\ref{fig:sim}, we report the price of strategyproofness: the difference in total similarity between the proposed algorithm's assignment and the optimal non-strategyproof assignment, as a fraction of the optimal assignment's  total similarity. Results for the random partitioning algorithm are averaged over $100$ trials; error bars representing standard error of the mean are too small to be visible. As expected from our theoretical results, the multi-partition algorithm achieves the full similarity of the optimal non-strategyproof assignment. On all values of $\assigndeg$, the cycle-breaking algorithm performs very well: it loses less than $1\%$ of the optimal similarity when $\assigndeg=1$, and furthermore, it outperforms the coloring algorithm even for higher values of $\assigndeg$ (where it does not have theoretical guarantees). The coloring algorithm loses around $12\%$ of the optimal similarity for all values of $\assigndeg$. The baseline of random partitioning still loses less than $20\%$ of the optimal similarity, but is outperformed by the other algorithms. Overall, on real data our algorithms  perform quite well in terms of the quality of the assignment as compared to the optimal non-strategyproof assignment.

\subsection{Partition Quality}
We next examine whether the partitions produced by these algorithms place similar-quality papers into each subset, since under the partition-based method, the final accept/reject decisions for papers are performed independently in each subset. 
In Figures~\ref{fig:outcome1} and \ref{fig:outcome3}, we display the number of papers receiving each decision (oral presentation, poster presentation, invitation to workshop track, or rejection) in each subset of the partitions. For each algorithm, each bar displays the decisions for the papers in one subset of the partition. Across all algorithms and values of $\assigndeg$, the partitions constructed have very similar numbers of papers receiving each decision in each subset. 
Since a very small number of papers (23 out of 883) are accepted for oral presentation overall, the relative difference in the number of oral papers between subsets is sometimes large; however, the absolute difference in the number of oral papers remains small. 

Further, in Figures~\ref{fig:score_cycle1}, \ref{fig:score_color1}, and \ref{fig:score_multi1}, we show the mean review scores given to each paper for the case of $\assigndeg=1$. In Figures~\ref{fig:score_cycle1} and \ref{fig:score_color1}, the red and blue histograms correspond to the scores given to the papers in the two subsets of the algorithm's partition, with the purple section indicating their overlap; in Figure~\ref{fig:score_multi1}, the third subset is additionally indicated in yellow. For all algorithms, the distributions of scores appear very similar across subsets of the partition. 
Formally, we test the difference between the score distributions of different subsets via the two-sample Kolmogorov-Smirnov test, a non-parametric test of the null hypothesis that the two samples came from the same distribution. Each sample is the set of scores given to the papers in one subset of the partition. 
We report the results of the test in Table~\ref{tab:ks}, which contains the $p$-values of the test along with the effect size $D$, defined as the maximum difference between the empirical cdfs of the two samples. For the multi-partition algorithm, we test each pair of subsets and we report results for the pair with highest $D$. In all cases, the $p$-values are high, meaning that the test cannot reject the hypothesis that the subsets were drawn from the same distribution. 

These experiments provide evidence that the partitions created by our algorithms do not have any substantial difference in the quality of papers in each subset. 

\begin{table}
\centering
\begin{tabular}{rlll} \toprule
Algorithm                                  & $\assigndeg$  & $p$      & $D$      \\ \midrule
Cycle-breaking                             & - & 0.9007 & 0.0373 \\
Coloring                                   & 1           & 0.8902 & 0.0379 \\
                                           & 2           & 0.6445 & 0.0487 \\
                                           & 3           & 0.5389 & 0.0530 \\
Multi-partition                            & 1           & 0.4282 & 0.0702 \\
                                           & 2           & 0.6805 & 0.0742 \\
                                           & 3           & 0.3457 & 0.1142\\
\bottomrule
\end{tabular} 
\caption{Results of the Kolmogorov-Smirnov test of whether the review scores in the two partitioned subsets are drawn from the same distribution.}\label{tab:ks} 
\end{table}

\begin{algorithm*}[t!]
\caption{Heuristic Algorithm for Arbitrary Authorship}
\label{algo:gen_auth}
\begin{algorithmic}[1]
\Require agents $\reviewers$, papers $\papers$, similarities $\simwhole$, authorship graph $\authorship$, paper load $\assigndeg_\adpap$, maximum agent load $\assigndeg_\adrev$
\State $\overline{\assignment}^* \gets$ max-similarity assignment with loads $(\assigndeg_\adrev, \assigndeg_\adpap)$ 
\State $\{V_1, \dots, V_N\} \gets$ vertices of the connected components of $\authorship$
\State $\reviewers' \gets \{\adrev_i' : i \in [N]\}$; $\papers' \gets \{\adpap_i' : i \in [N]\}$
\For {$i,j \in [N]$}
\State $\similarity'_{ij} \gets \sum_{\adrev_a \in V_i, \adpap_b \in V_j} \similarity_{ab} \mathbb{I}[(\adrev_a, \adpap_b) \in \overline{\assignment}^*] + \sum_{\adrev_a \in V_j, \adpap_b \in V_i} \similarity_{ab} \mathbb{I}[(\adrev_a, \adpap_b) \in \overline{\assignment}^*]$
\EndFor
\State $\assignment', (\reviewers_1', \reviewers_2') \gets$ output of Algorithm~\ref{algo:deg1} on input $(\reviewers', \papers', \simwhole', \assigndeg'=1)$ \label{ln:alg_k1_call}
\State $\mathcal{T}_1 \gets \bigcup_{i : \adrev_i' \in \reviewers_1'} V_i$;  $\mathcal{T}_2 \gets \bigcup_{i : \adrev_i' \in \reviewers_2'} V_i$
\State $\assignment \gets $ max-similarity assignment with loads $(\assigndeg_\adrev, \assigndeg_\adpap)$  respecting $(\mathcal{T}_1, \mathcal{T}_2)$
\State \Return assignment $\assignment$ and partition $(\mathcal{T}_1, \mathcal{T}_2)$
\end{algorithmic}
\end{algorithm*}

\section{Heuristic Algorithm for Arbitrary Authorship} \label{sec:gen_auth}
In this section, we propose an algorithm for \strategyproof{} assignment that can accommodate arbitrary authorship of submissions, as opposed to the one-to-one authorship that we assume in our problem formulation (Section~\ref{sec:prob}). This algorithm is closely based on the cycle-breaking algorithm (Algorithm~\ref{algo:deg1}) from Section~\ref{sec:deg1}. We do not have any theoretical guarantees for this algorithm, but we provide evaluations on the ICLR 2018 dataset introduced in Section~\ref{sec:exps}.

\subsection{Algorithm}
Arbitrary authorship can be represented as a graph $\authorship$ where each agent and each submission are represented as vertices, and an edge between an agent and submission indicates that the agent authored that submission. Since authorship is not one-to-one, the number of agents and submissions may differ and the agent and submission loads need not be the same. Define $\assigndeg_\adpap$ as the paper load and $\assigndeg_\adrev$ as the maximum agent load. A \strategyproof{} assignment algorithm in this setting will produce a partition of both agents and submissions, along with an assignment that respects this partition by assigning each submission only agents from the other subset.

\begin{figure*}[t!] 
    \centering
    \begin{subfigure}{0.24\textwidth}\includegraphics[width=1\textwidth]{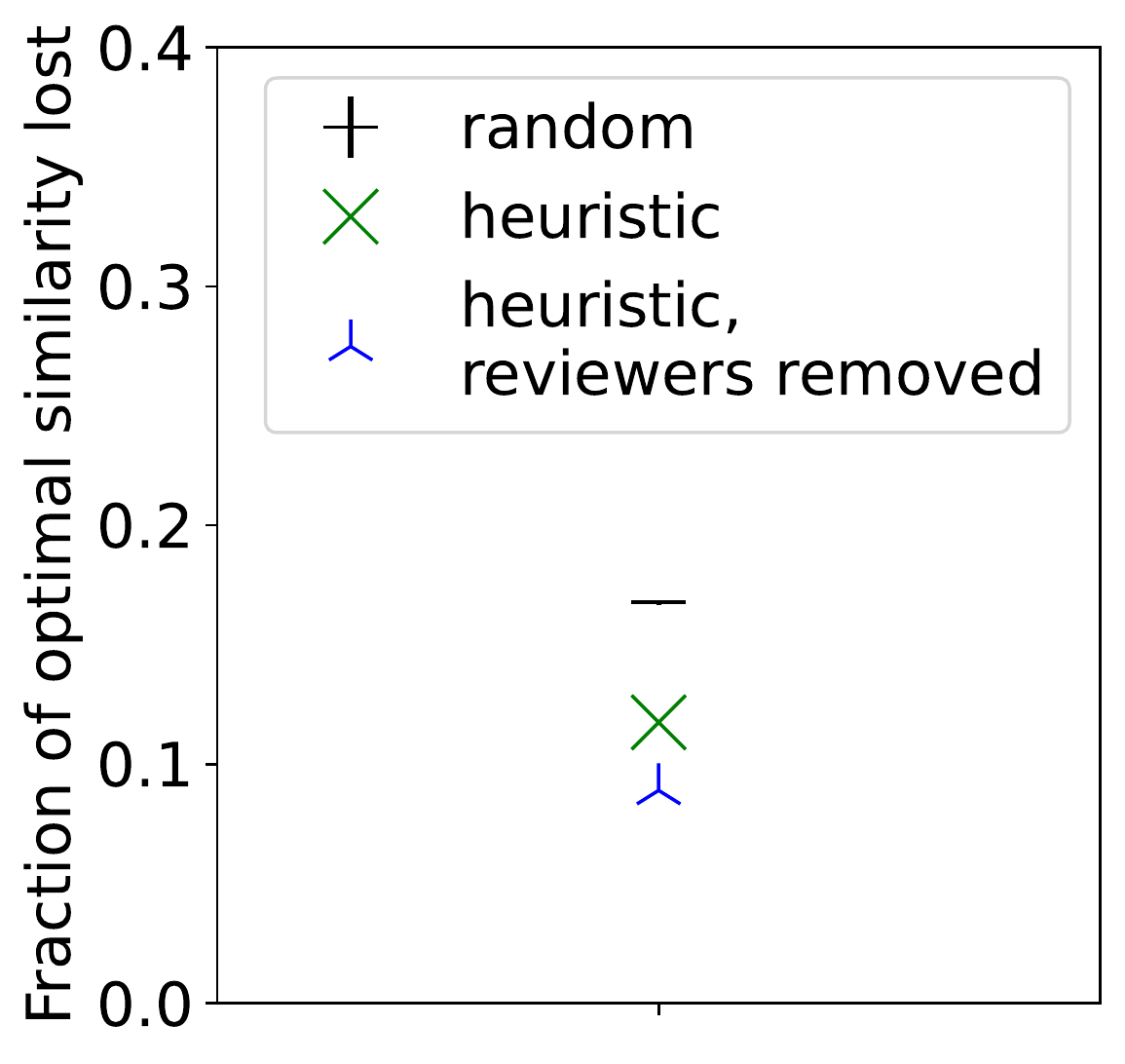}\caption{Assignment similarity lost}\label{fig:sim_gen}
    \end{subfigure}
    \begin{subfigure}{0.32\textwidth}\includegraphics[width=1\textwidth]{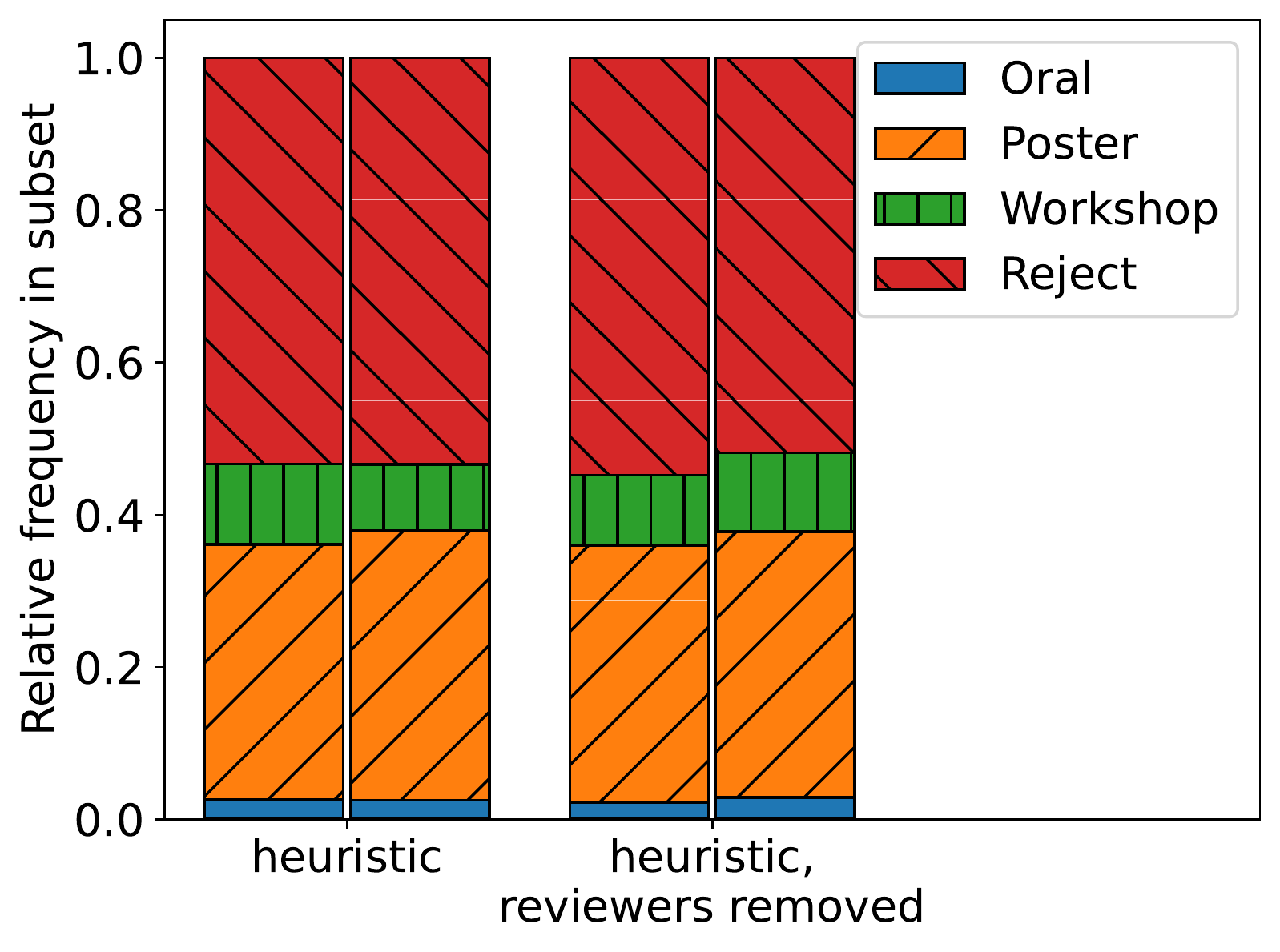}\caption{Partitioned paper decisions}\label{fig:outcome_gen} \end{subfigure} 
    \begin{subfigure}{0.32\textwidth}\includegraphics[width=1\textwidth]{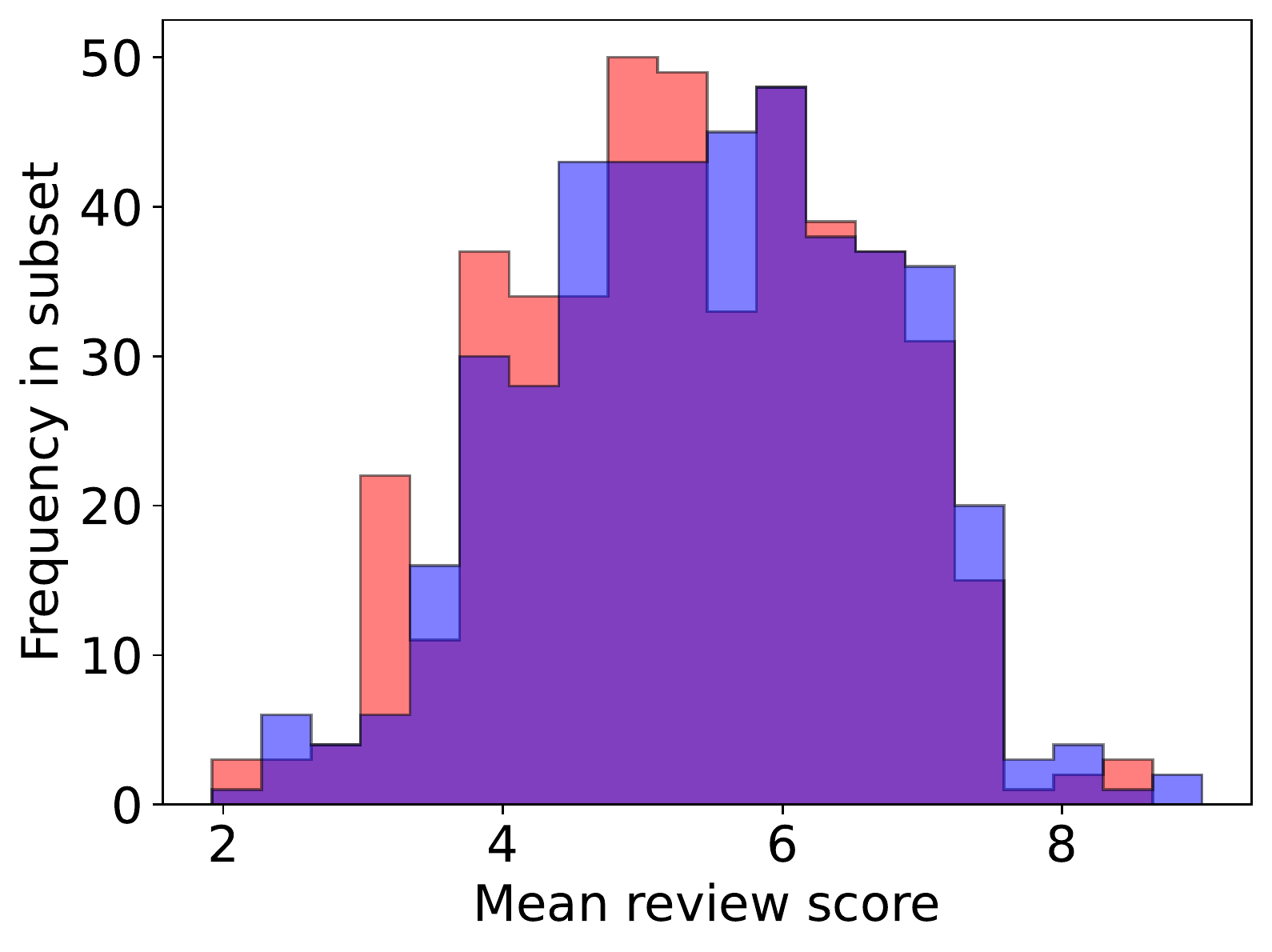}\caption{Partitioned paper scores, heuristic algorithm with reviewers removed}\label{fig:score_gen} \end{subfigure}
    \caption{Experimental results using Algorithm~\ref{algo:gen_auth} on the authorship from ICLR 2018.} \label{fig:results_gen} 
\end{figure*}

Algorithm~\ref{algo:gen_auth} works by taking a problem instance with arbitrary authorship, using it to construct a (fake) problem instance with one-to-one authorship, and running Algorithm~\ref{algo:deg1} on this fake instance to find a partition. Each agent in the fake instance corresponds to a connected component of the authorship graph $\authorship$. Similarities between fake agents are set equal to the total similarity of pairs in the optimal non-strategyproof assignment that are split between the respective components. After construction, we pass this fake instance into Algorithm~\ref{algo:deg1}.

We slightly modify Algorithm~\ref{algo:deg1} to encourage more balanced partitions in this setting before calling it in Line~\ref{ln:alg_k1_call}. In Lines~\ref{ln:add_to_set}-\ref{ln:end_add} of Algorithm~\ref{algo:deg1}, we take the larger of $A$ and $B$ and add it to the smaller of $\reviewers_1$ and $\reviewers_2$ as measured by the total number of papers in the connected components represented within each set. In addition, we iterate through vertices (when finding cycles) in the order of largest connected component to smallest, where size is again determined by the number of papers in each component.

\subsection{Experimental Results}
We test Algorithm~\ref{algo:gen_auth} on the ICLR 2018 dataset, using the full authorship graph from the conference. Following the suggestion in~\cite{xu2018strategyproof}, we also try running Algorithm~\ref{algo:gen_auth} after removing reviewers with a large number of authored papers; this breaks up large connected components in the authorship graph, thus allowing more flexibility in choosing a partition. Specifically, we remove the $53$ reviewers with more than $3$ papers authored ($2.2\%$ of reviewers) from the reviewer pool. 
As a baseline for comparison, we also test $100$ trials of random partitioning, which chooses half of the connected components at random for each subset. We set loads of $\assigndeg_\adpap = 3$ and $\assigndeg_\adrev = 6$, since these are standard conference loads~\cite{xu2018strategyproof}.

First, we see in Figure~\ref{fig:sim_gen} that Algorithm~\ref{algo:gen_auth} outperforms random partitioning in terms of similarity. Our algorithm loses $11.7\%$ of the non-strategyproof optimal similarity, whereas the random partitioning loses $16.8\%$ of optimal on average. When we remove high-authorship reviewers before running Algorithm~\ref{algo:gen_auth}, it only loses $8.9\%$ of the optimal similarity (which is still allowed to use all reviewers). 

Finally, we examine the partition quality in a similar manner as in Section~\ref{sec:exps}.  In Figure~\ref{fig:outcome_gen}, we plot the proportion of papers within each subset of the partitions produced by Algorithm~\ref{algo:gen_auth} that received each decision. We see that the subsets have similar proportions of papers receiving each decision, regardless of whether we remove high-authorship reviewers. However, removing these reviewers results in a significantly more balanced partition: the number of papers differs between subsets by $109$ when high-authorship reviewers are not removed and by only $1$ when they are. In Figure~\ref{fig:score_gen}, we see that the two subsets also have similar distributions of paper scores when high-authorship reviewers are removed.

Our results are highly comparable to those of~\cite{xu2018strategyproof}, who provide a partitioning algorithm that simply returns an arbitrary feasible partition of the connected components of the authorship graph. The authors report that this algorithm loses only $11.4\%$ of the optimal similarity on the ICLR 2018 data with the same loads, a similar performance to our algorithm's despite the fact that our algorithm more carefully chooses the partition. This phenomenon may be related to the results of~\cite{jecmen2022nearoptimal}, who find that randomly splitting reviewers into two ``phases'' of reviewing does not significantly degrade assignment quality on real conference datasets.

\fullver{

\section{Discussion}

We jointly considered two key aspects of the peer-assessment process---strategyproofing and assignment quality---and derived fundamental limits as well as designed computationally-efficient algorithms that achieve these limits. Our theoretical and empirical contributions lead to several directions of future work.

A first key direction of future work is to extend these theoretical results to arbitrary authorship graphs, as in conference peer review. We present a heuristic algorithm with an empirical evaluation in Section~\ref{sec:gen_auth}, but the problem of establishing fundamental limits and optimal algorithms is open. 
Second, most of our work considered worst-case guarantees, while showing that it is NP-hard to attain instance-wise optimality. However, our experimental results showed that our algorithms perform much better than worst-case on real-world instances. This suggests a theoretically interesting and practically useful direction of future work: designing algorithms with approximately-optimal instance-wise guarantees. Third, in contrast to past work, our partitions are non-random. Building on our experimental results revealing that these non-random partitions still result in subsets with roughly equal submission strengths, future work could dig deeper into this phenomenon both theoretically and empirically. Fourth, recent work~\cite{mattei2020peernomination} provides a strategyproof algorithm with theoretical guarantees that does not rely on partitioning. Even though partitioning is by far the dominant way of strategyproofing, it is of interest to extend our results to such strategyproofing methods that may not employ partitioning. 
Finally, there are various other types of strategic or dishonest behavior in peer assessment~\cite{stelmakh2020catch,hvistendahl2013china,ferguson2014publishing,fanelli2009many,resnik2008perceptions,Vijaykumar2020Architecture,littman2021collusion,jecmen2020manipulation,wu2021making} and the design of computational methods to mitigate such behavior is vital. More generally, peer assessment is an important application with a broad set of challenges including subjectivity~\cite{lee2015commensuration,noothigattu2020loss}, miscalibration~\cite{roos2011calibrate,wang2018your}, biases~\cite{tomkins2017reviewer,stelmakh2019testing,manzoor2020uncovering}, and others~\cite{meir2020market,fiez2020super,stelmakh2020resubmissions,wang2020debiasing,shah2021survey}.

\section*{Acknowledgments}
This work was supported by NSF CAREER award 1942124 and a Google Research Scholar Award. We thank Ariel Procaccia for helpful discussions.
}

\ifisarxiv

\bibliographystyle{unsrt}
{\small
\bibliography{bibtex}}

\else 

\bibliography{bibtex} 
\clearpage
\pagebreak
\end{document}
\fi

~\\
\appendix

\noindent{\LARGE \bf Appendices}

\confver{
\section{Fairness Objective} \label{sec:minimax} 
\secfairness
}

\confver{
\section{Omitted Proofs}
In this section we present proofs that could not be included in the main text due to space constraints.

\subsection{Proof of Proposition~\ref{prop:rand}} \label{apdx:rand}
\proofproprand

\subsection{Proof of Theorem~\ref{thm:nph}} }
\fullver{\section{Proof of Theorem~\ref{thm:nph}} } \label{apdx:nph}
We first prove the following supplementary lemma. 
\begin{lemma} \label{lem:graphorient}
Any graph $G = (V, E)$ with maximum degree at most $4$ can be oriented in polynomial time such that the in-degree and out-degree of all vertices are at most $2$. 
\end{lemma}
\begin{proof}
Consider the following procedure, which takes any arbitrary orientation of $G$ and modifies it that it obeys the desired properties. For a vertex $v \in V$, denote the out-degree of $v$ by $\delta_o(v)$ and the in-degree of $v$ by $\delta_i(v)$.

On each iteration, choose $v_o$ such that $\delta_o(v_o) \geq 3$. Consider the set $\mathcal{S}$ of all vertices reachable on a directed path from $v_o$. There cannot be any edges from $\mathcal{S}$ to vertices outside of $\mathcal{S}$, so $\sum_{v \in \mathcal{S}} \delta_o(v) \leq \sum_{v \in \mathcal{S}} \delta_i(v)$. Since $\sum_{v \in \mathcal{S}} \delta_i(v) + \delta_o(v) \leq 4 |\mathcal{S}|$, $\sum_{v \in \mathcal{S}} \delta_o(v) \leq 2 |\mathcal{S}|$. Therefore, there exists $v' \in \mathcal{S}$ such that $\delta_o(v') \leq 1$. Reversing the direction of all edges on the path from $v_o$ to $v'$ reduces $\delta_o(v_o)$ by $1$ and increases $\delta_i(v_o)$ by $1$, increases $\delta_o(v')$ by $1$ and decreases $\delta_i(v')$ by $1$, and does not change the in- or out-degree of any other vertices. Since $\delta_o(v') \leq 1$ and $\delta_i(v_o) \leq 1$ originally, this change does not cause any additional constraints to be violated. Therefore, this step can be repeated until all vertices satisfy $\delta_o(v) \leq 2$. 
Reversing the direction of all edges and repeating the entire procedure ensures that all vertices satisfy $\delta_i(v) \leq 2$ as well.

Each iteration takes $O(|V|)$ time to find a path to an appropriate $v'$. The number of iterations is $O(|V|)$, since each vertex is $v_o$ at most twice.
\end{proof}

We now prove the main result, showing that it is unlikely that an instance-optimal algorithm for \strategyproof{} assignment exists if $\assigndeg \geq 2$.
We reduce from the ``Simple Max Cut on Cubic Graphs'' problem, which is NP-complete~\cite{yannakakis1978node}. An instance of this problem consists of an unweighted, undirected graph $G = (V, E)$ where all vertices have degree $3$ and an integer $K$. The question is: is there a partition of $V$ that cuts at least $K$ edges? 

Fix any $\assigndeg \geq 2$. We reduce this problem to a decision variant of our problem, defined as follows. Given agents $\reviewers$, submissions $\papers$, similarities $\simwhole$, and a number $x$, does there exist a \strategyproof{} assignment with loads of $\assigndeg$ such that the total similarity is at least $x$?
If it is NP-hard to determine if there is a \strategyproof{} assignment with similarity at least $x$, finding the \strategyproof{} assignment with maximum similarity must also be NP-hard.

Consider any instance of the max cut problem $G = (V, E)$, $K$. Construct a graph $G'$ by adding $|V|$ disconnected vertices to $G$. By Lemma~\ref{lem:graphorient}, we can find an orientation of $G'$ in polynomial time to get a directed graph $\hat{G} = (\hat{V}, \hat{E})$ such that all out-degrees and in-degrees are at most $2$. For each vertex $v_i \in V$, construct one agent $\adrev_i$ and their submission $\adpap_i$. For each directed edge $(v_i, v_j) \in \hat{E}$, set similarity $S_{ij} = 1$; set all other similarities to $0$. Set $x = K$.

Suppose that $V$ has a partition $(V_1, V_2)$ that cuts at least $K$ edges. Add the disconnected vertices to the subsets so that $|V_1| = |V_2| = |V|$. Partition the corresponding agents in the same way. Assign the agent $\adrev_i$ to submission $\adpap_j$ for each directed edge $(v_i, v_j) \in \hat{E}$ cut by the partition. Each agent has an edge to at most $2$ submissions and each submission has an edge to at most $2$ agents, so these can all be assigned since $k \geq 2$. Assign the remaining agents and submissions arbitrarily, which can be done since the partitions are balanced. This assignment has similarity at least $x$ and is \strategyproof{}.

Suppose that $V$ does not have a partition that cuts at least $K$ edges. This means that there does not exist a partition of agents such that at least $x$ agent-submission pairs with non-zero similarity can be assigned respecting the partition. 

\confver{
\subsection{Proof of Theorem~\ref{thm:kpart}} \label{apdx:kpart}
Each vertex in $\dirsimilarity$ has in-degree and out-degree $\assigndeg$, so the maximum (total) degree is at most $2\assigndeg$. Therefore, by Theorem~\ref{thm:hajnal} we can find an equitable $(2\assigndeg + 1)$-coloring of $\dirsimilarity$ in $O(\numrev^2)$ time. By Definition~\ref{def:eq_color}, the entirety of $\optassign_\simwhole$ respects the partition induced by the coloring and so is \strategyproof{} with respect to this partition. Also by Definition~\ref{def:eq_color}, all color classes differ in size by at most $1$.
}

\begin{figure*}[t!] 
    \centering
    \begin{subfigure}{0.32\textwidth}\includegraphics[width=1\textwidth]{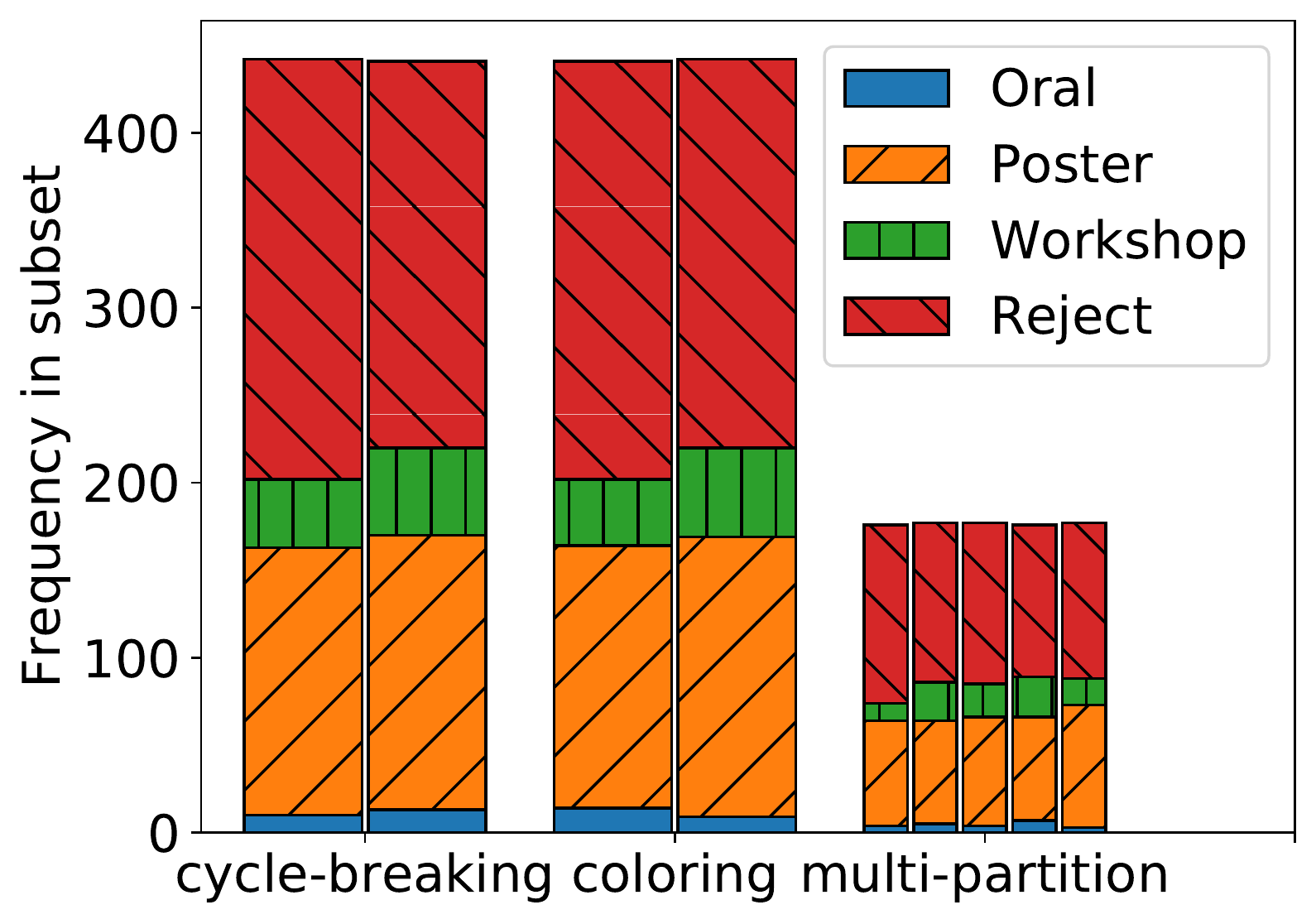}\caption{Partitioned paper decisions, \\ $\assigndeg=2$}\label{fig:outcome2} \end{subfigure} 
    \begin{subfigure}{0.32\textwidth}\includegraphics[width=1\textwidth]{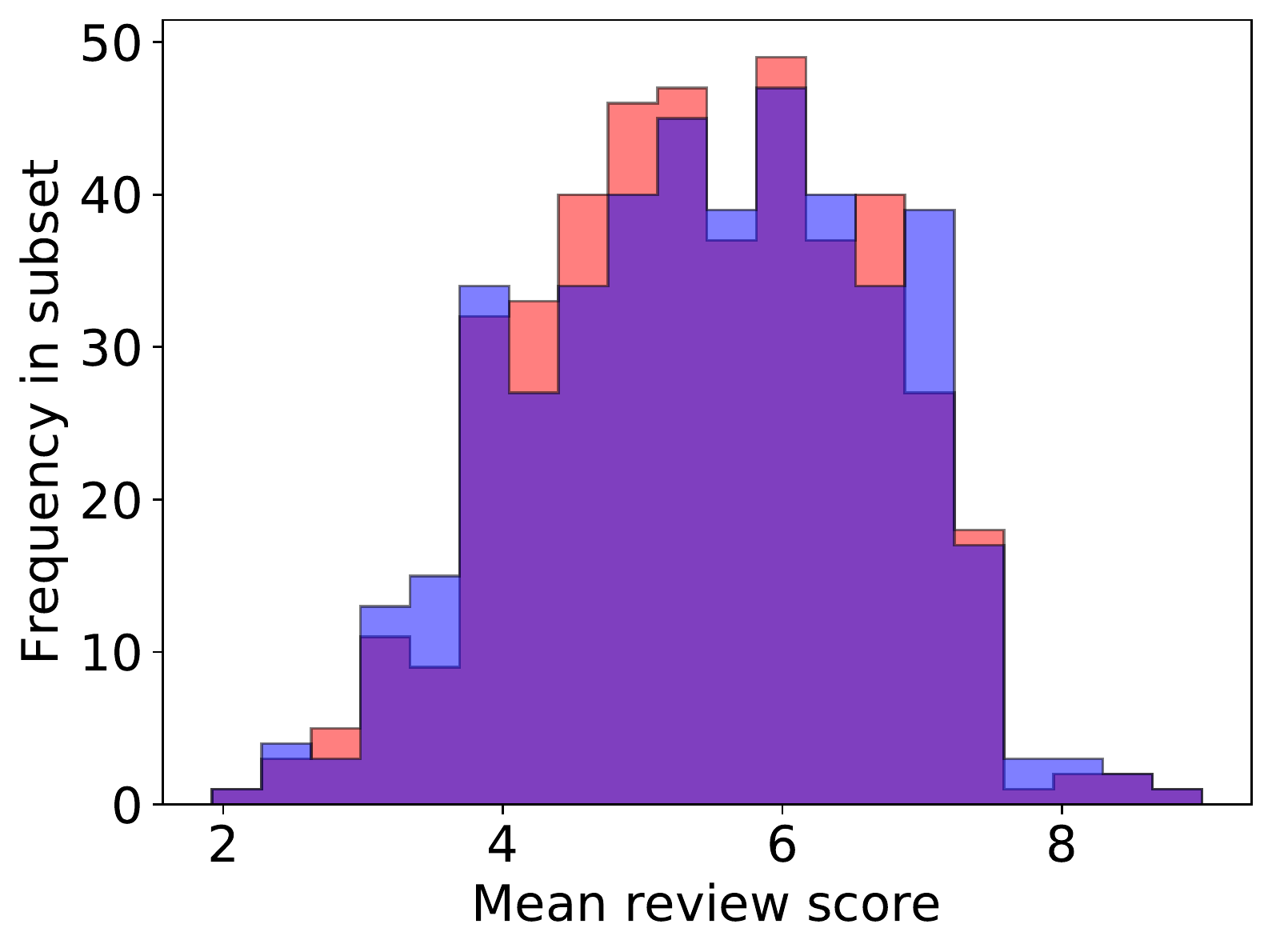}\caption{Partitioned paper scores for \\ the cycle-breaking algorithm, $\assigndeg=2$}\label{fig:score_cycle2} \end{subfigure}
    \begin{subfigure}{0.32\textwidth}\includegraphics[width=1\textwidth]{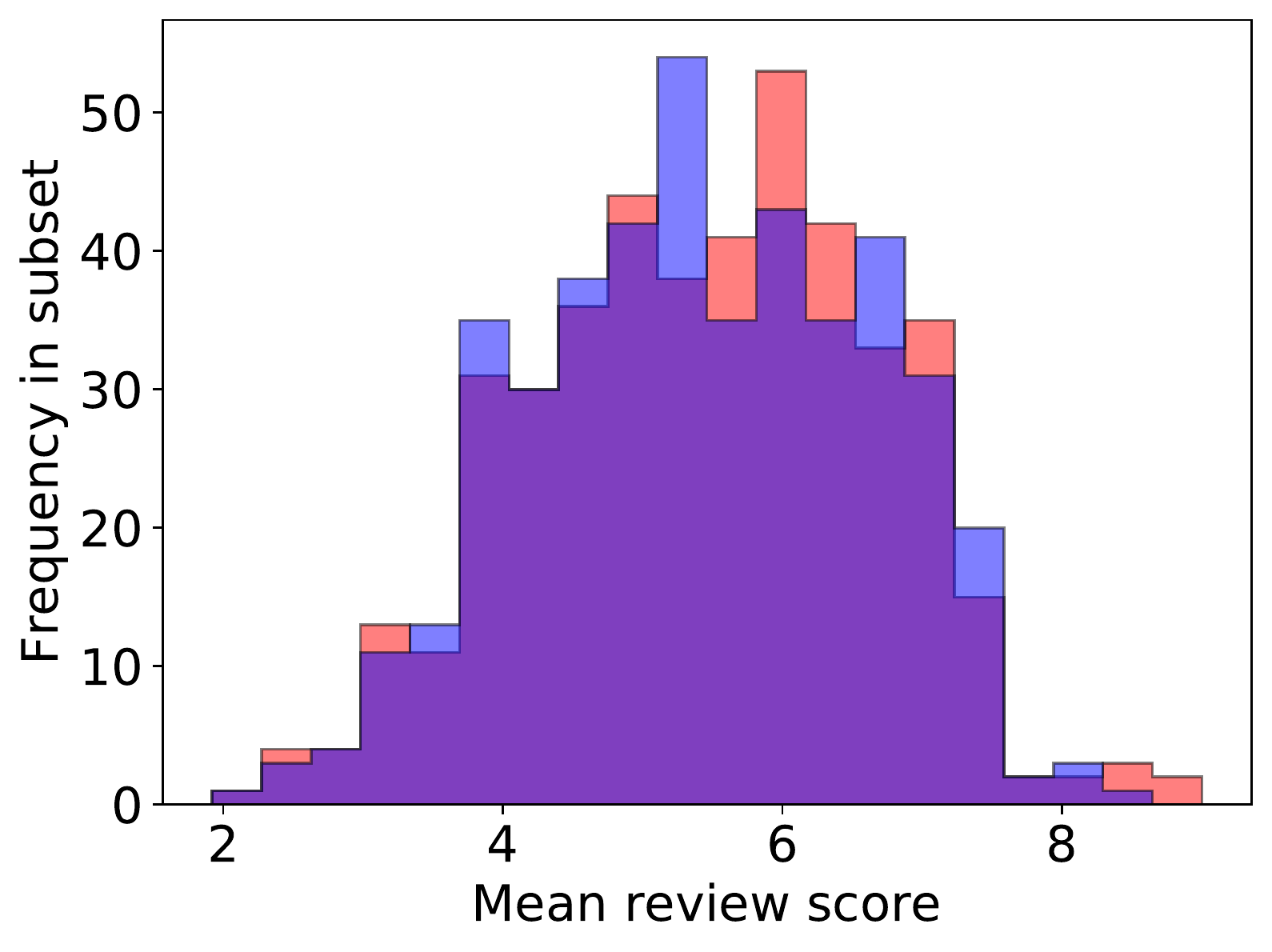}\caption{Partitioned paper scores for \\ the coloring algorithm, $\assigndeg=2$}\label{fig:score_color2} \end{subfigure}  \\
    \begin{subfigure}{0.32\textwidth}\includegraphics[width=1\textwidth]{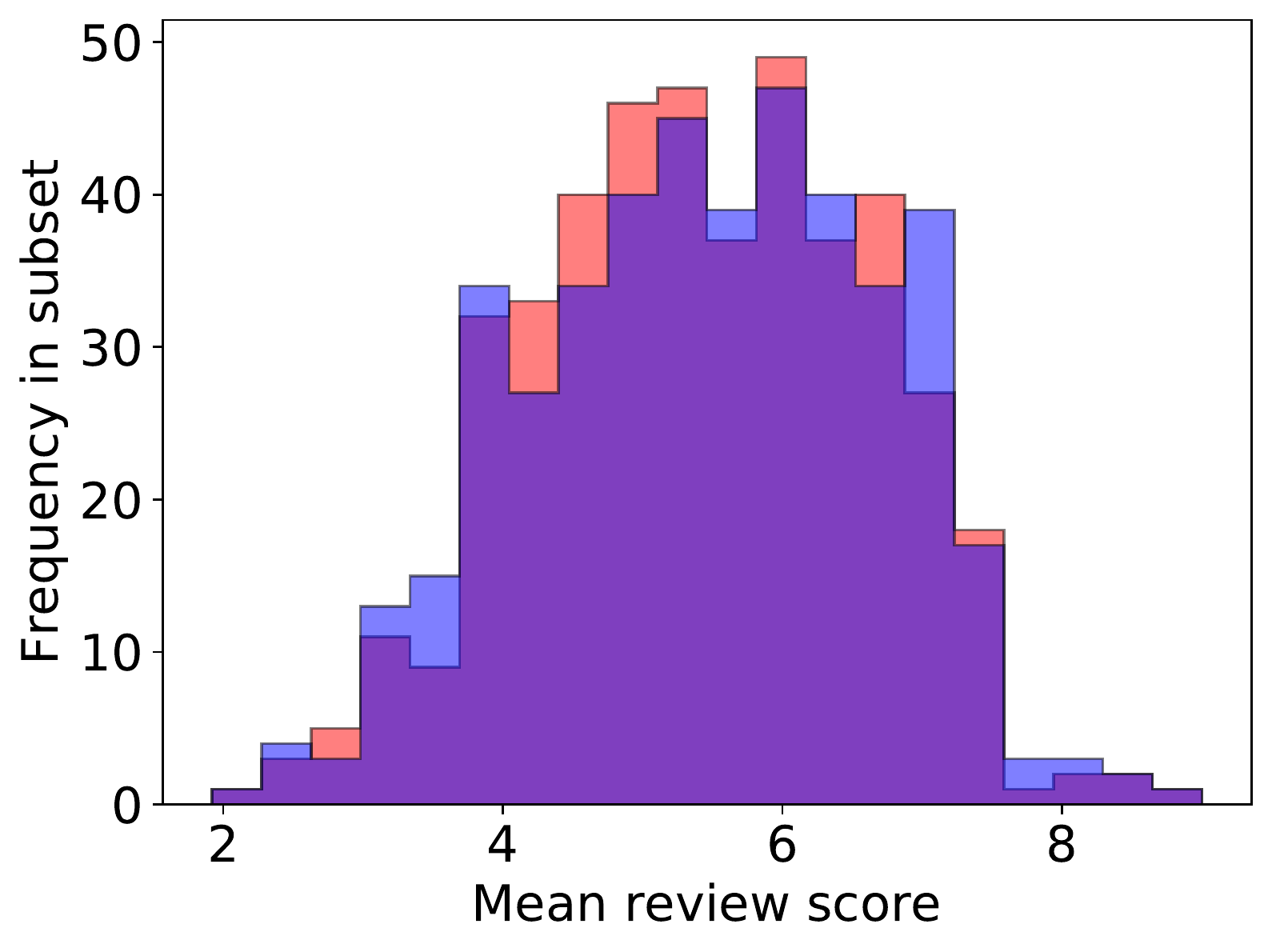}\caption{Partitioned paper scores for \\ the cycle-breaking algorithm, $\assigndeg=3$}\label{fig:score_cycle3} \end{subfigure}
    \begin{subfigure}{0.32\textwidth}\includegraphics[width=1\textwidth]{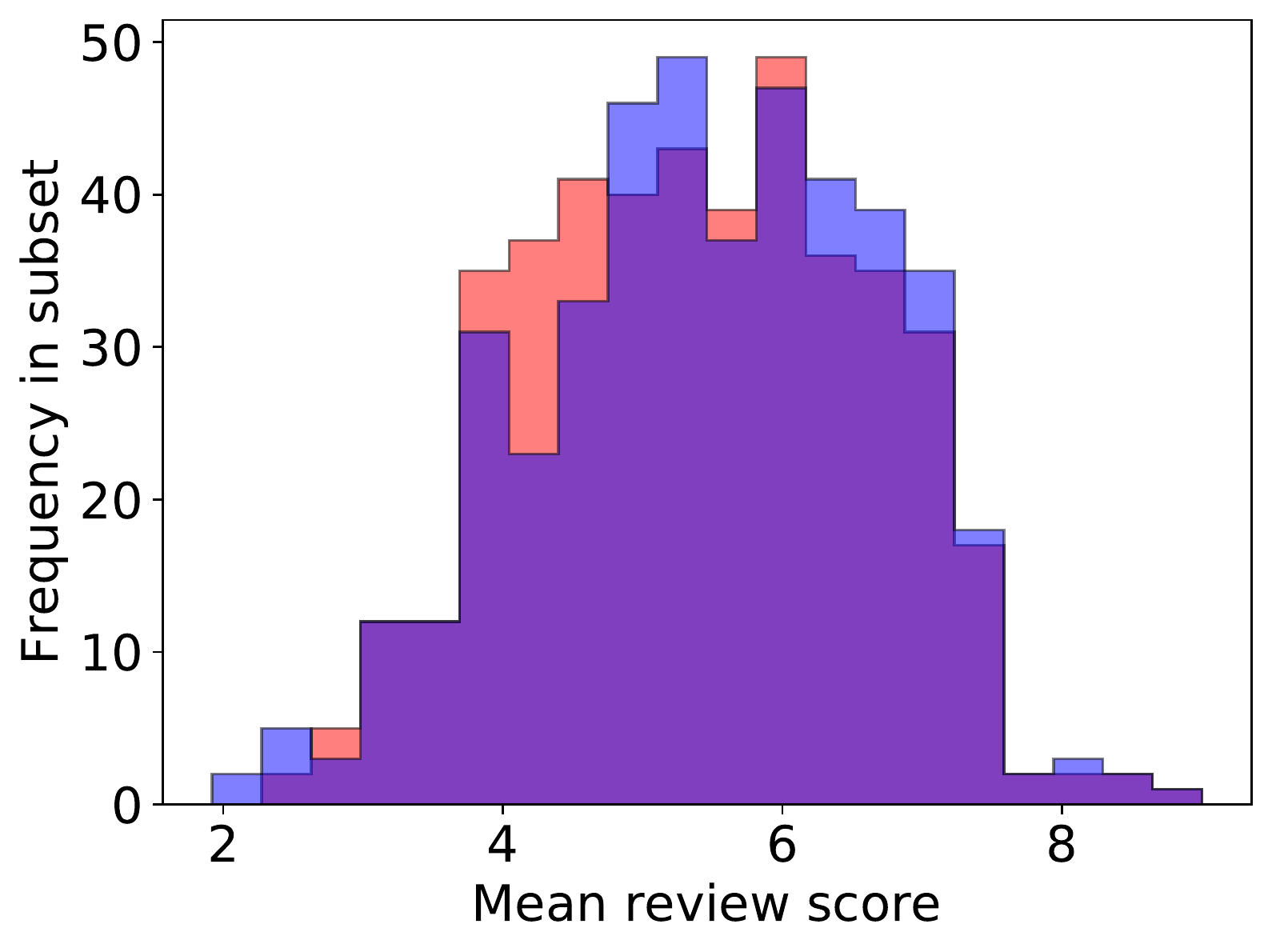}\caption{Partitioned paper scores for \\ the coloring algorithm, $\assigndeg=3$}\label{fig:score_color3} \end{subfigure} 
    \caption{Additional experimental results on data from ICLR 2018.} \label{fig:results2} 
\end{figure*}

\section{Additional Experimental Results} \label{apdx:exps}
In this section, we display additional experimental results regarding the partition quality of our algorithms on the data from ICLR 2018 (introduced in Section~\ref{sec:exps}). 

In Figure~\ref{fig:outcome2}, we display the number of papers receiving each decision in each subset of the partitions for $\assigndeg=2$, where each bar displays the decisions for the papers in one subset of the partition. The partitions constructed by all algorithms have very similar numbers of papers receiving each decision in each subset.

In Figures~\ref{fig:score_cycle2}-\ref{fig:score_color3}, we show the mean review scores given to each paper for the cases of $\assigndeg=2$ and $\assigndeg=3$. As before, the red and blue histograms correspond to the scores given to the papers in each subset of the algorithm's partition, with the purple section indicating their overlap. For all algorithms, the distribution of scores appear very similar across subsets of the partition. Results for the multi-partition algorithm are not shown, as there are too many subsets for the histogram to be readable.

\end{document}